\documentclass[twoside,a4paper]{article}
\usepackage{amsmath,graphicx,amssymb,fancyhdr,amsthm,,textcomp,hyperref}

\usepackage[english]{babel}
\usepackage[numbers]{natbib}
\usepackage[acronym, nopostdot,nomain,nonumberlist,numberedsection]{glossaries}

\makenoidxglossaries

\newtheorem{thm}{Theorem}[section] 
\newtheorem{cor}[thm]{Corollary} 
 
\newtheorem{prop}[thm]{Proposition} 
 
\theoremstyle{definition} 
\newtheorem{defn}[thm]{Definition}  
\theoremstyle{remark}  
  
\def\beq{\begin{eqnarray}}  
\def\eeq{\end{eqnarray}}  
\def\bsp{\begin{split}}  
\def\esp{\end{split}}

\def\Tr{\mathrm{Tr}}  
\def\d{\mathrm{d}}

\newacronym{gr}{GR}{General Relativity}
\newacronym[plural=THs,firstplural=trapping horizons (THs)]{ths}{TH}{trapping horizon}
\newacronym[plural=MTTs,firstplural=marginally trapped tubes (MTTs)]{mtts}{MTT}{marginally trapped tube}
\newacronym{neh}{NEH}{non-expanding horizon}
\newacronym{wih}{WIH}{weakly isolated horizon}
\newacronym[plural=MTSs,firstplural=marginally trapped surfaces (MTSs)]{mtss}{MTS}{marginally trapped surface}
\newacronym{foth}{FOTH}{future outer trapping horizon}
\newacronym{dh}{DH}{dynamical horizon}
\newacronym{bw}{b.w.}{boost weight}
\newacronym[plural=SPIs,firstplural=scalar polynomial (curvature) invariants (SPIs)]{spi}{SPI}{scalar polynomial (curvature) invariant}
\newacronym{np}{NP}{Newman-Penrose}
\newacronym{fkwc}{FWKC}{Fulling, King, Wybourne and Cummings}
\newacronym{ltb}{LTB}{Lemaitre-Tolman-Bondi}
\newacronym{nut}{NUT}{Newman, Unti, Tamburino}
\begin{document}  
  
\title{\Large\textbf{Identification of black hole horizons using scalar curvature invariants}}  
\author{{\large\textbf{Alan Coley$^{\heartsuit}$ and David McNutt$^{\spadesuit}$ }}
%EndAName  
%\address{  
\vspace{0.3cm} \\ 
$^{\heartsuit}$Department of Mathematics and Statistics,\\
Dalhousie University,
Halifax, Nova Scotia,\\
Canada B3H 3J5
%\email{  
\vspace{0.3cm} \\ 
$^{\spadesuit}$ Faculty of Science and Technology,\\
University of Stavanger, 
N-4036 Stavanger, Norway  \\
 \vspace{0.3cm} \\
\texttt{aac@mathstat.dal.ca,david.d.mcnutt@uis.no}}  
\date{\today}  
\maketitle  
\pagestyle{fancy}  
\fancyhead{} % clear all header fields  
\fancyhead[EC]{}  
\fancyhead[EL,OR]{\thepage}  
\fancyhead[OC]{}  
\fancyfoot{} % clear all footer fields  

\begin{abstract}

We introduce the concept of a geometric horizon, which is a surface distinguished by the vanishing of certain curvature invariants which characterize its special algebraic character. We motivate its use for the detection of the event horizon of a stationary black hole by providing a set of appropriate scalar polynomial curvature invariants that vanish on this surface. We extend this result by proving that a non-expanding horizon, which generalizes a Killing horizon, coincides with the geometric horizon. Finally, we consider the imploding spherically symmetric metrics and show that the geometric horizon identifies a unique quasi-local surface corresponding to the unique spherically symmetric marginally trapped tube, implying that the spherically symmetric dynamical black holes admit a geometric horizon. Based on these results, we propose a suite of conjectures concerning the application of geometric horizons to more general dynamical black hole scenarios.

\end{abstract} 
%%%%%%%%%%%%%%%%%%%%%%%%%%%%%%%%%%%%%%%%%%%%%

\newpage 

%\tableofcontents
%
%\newpage

\section{Introduction}

%Black holes  are exact solutions in \gls{gr} which are presumed to represent physical objects resulting from the gravitational collapse of fuel-exhausted stars; they exhibit a deep relationship between gravitation, thermodynamics, and quantum theory. A characteristic feature of a black hole is its event horizon:  the boundary of the region from where signals can be sent to a distant asymptotic external region.  
In General Relativity (GR), black holes are exact solutions which may be interpreted as physical objects formed from the gravitational collapse of fuel-exhausted stars. As such, they present an excellent arena to explore the connection between gravitation, thermodynamics and quantum theory. A defining feature of a black hole is its event horizon, which is the boundary of the region from where signals can be sent to a distant asymptotic external region.
The event horizon is typically identified as the surface of the black hole and relates its area to the entropy of the black hole. However, the event horizon is essentially a {\em teleological} object, as we must know the global behaviour of the spacetime in order to determine the event horizon locally. That is, the event horizon depends on the {whole} future evolution of the spacetime \cite{AshtekarKrishnan}. 

To study the behaviour of black holes, for example in  numerical GR \cite{T} in the 3+1 approach  or in the Cauchy-problem in GR, it is crucial to locate a black hole locally. Of course, such a characterization may not rely on the existence of an event horizon alone, as realistic black holes undergo evolutionary processes and are usually dynamical. To address this, Penrose \cite{P2} introduced the important concept of {\em closed
trapped surfaces}, which are compact spacelike surfaces (usually
topological spheres) such that the expansion of the future-pointing null normal vectors are negative. Considering 
time-dependent situations, the event horizons (which are null surfaces) familiar from the study of stationary black holes are replaced in practice by {\it apparent horizons} defined as the locus of the vanishing expansion of a null geodesic congruence emanating from a trapped surface $S$ with spherical topology \cite{Booth2005}.

Unlike the event horizon, the apparent horizon is a quasi-local concept, and it is intrinsically foliation-dependent; this is  because it is a 2-surface that is dependent on the choice of asymptotically flat 3-surfaces which foliate spacetime, and consequently will depend on the observer in dynamical situations. In numerical studies of collapse, it is more practical to track apparent horizons because, as already noted, the event horizon requires the knowledge of the entire future history of the spacetime.  For example, apparent horizons are employed in simulations of high precision waveforms of gravitational waves arising from the merger of compact-object binary systems or in stellar collapse to form black holes in numerical relativity. The successful observations by the LIGO collaboration of gravitational waves from black hole mergers relied upon numerical simulations based on apparent horizons \cite{LIGO}. 

In practice the definition of an apparent horizon is difficult to use, and other quasi-local surfaces are often employed instead. In particular, two quasi-local surfaces, \glspl{mtts} and \glspl{ths}, which bound the event horizon of a dynamical black hole, play an important role  \cite{Senov}. These surfaces are extensions of the concept of a {\it future-trapped surface}, $S$, which is a closed two-surface with the property that the expansions in each of the two future-pointing null vectors normal to the surface, $\ell^a$ and $n^a$, are everywhere negative: 

\beq \theta_{(\ell)} = \bar{q}^{ab} \nabla_a \ell_b < 0,~~\text{ and } \theta_{(n)} = \bar{q}^{ab} \nabla_a n_b <0 \label{eqn:trappedS} \eeq

\noindent where $\bar{q}_{ab} = g_{ab} + \ell_a n _b + \ell_b n_a$ is the induced two metric on $S$. We will always normalize the null vectors such that $\ell_a n^a = -1$ to ensure they are outward/inward pointing null vector fields. As an example of such a surface, for equilibrium states of  dynamical black holes an alternative to the restrictive concept of a stationary horizon is given by the quasi-local weakly isolated horizons, which account for equilibrium states of black holes and cover all essential local features of event horizons which are unaffected by the dynamic evolution of the surrounding spacetime \cite{AshtekarKrishnan,DiazPolo}. 

\begin{defn}
A sub-manifold $\mathcal{H}$ of a spacetime is said to be a \gls{neh} if 
\begin{itemize}
\item $\mathcal{H}$ is topologically $S^2 \times \mathcal{R}$ and null.
\item Any null normal $\ell^a$ of $\mathcal{H}$ has vanishing expansion $\theta_{(\ell)} = 0$ and
\item The Einstein field equations hold at $\mathcal{H}$ and the stress-energy tensor $T_{ab}$ is such that $-T^a_{~b} \ell^b $ is future-causal for any future directed null normal $\ell^a$.
\end{itemize}
The pair $(\mathcal{H}, \ell)$ is said to constitute a \gls{wih} provided $\mathcal{H}$ is an NEH and any null normal proportional to $\ell^a$ satisfies
\beq (\mathcal{L}_\ell \mathcal{D}_a - \mathcal{D}_a \mathcal{L}_\ell) \ell^b =0 \nonumber \eeq
\noindent where $\mathcal{D}_a$ is the induced torsion-free derivative operator on $\mathcal{H}$. For any NEH with the condition that $\theta_{(\ell)} = 0$ on $\mathcal{H}$ implies that $\mathcal{D}_a$ will be unique \cite{Ashtekar}.

\end{defn} 

\noindent A WIH is essentially a three-dimensional (3D) null surface with topology $S^2 \times R$ with an outgoing expansion rate which vanishes on the horizon (with some additional conditions) \cite{AshtekarKrishnan}. The null normal vector is a local time-translational Killing vector field for the intrinsic geometry of the horizon, leading to the invariance under time evolution with respect to $\ell^a$ of the induced metric and the induced derivative operator, which is directly expressible by the isolated horizon condition. All these conditions are local to the horizon,  and require neither asymptotic structures nor foliations of spacetime. Every Killing horizon which is topologically $S^2 \times R$  is an isolated horizon. However, in general, spacetimes with isolated horizons need not admit any Killing vector fields for the spacetime, even in a neighborhood. 

A {\it trapped region} $\mathcal{T}$ is defined as a subset of spacetime where each point of this region passes a trapped surface. The {\it trapping boundary} $\partial \mathcal{T}$ is a connected component of the boundary of an inextendible trapped region. Unlike the MTTs, the trapping boundary is not foliated by \glspl{mtss} which are compact spacelike two-dimensional (2D) submanifolds on which the expansion of one of the null normals vanishes, and the other is non-positive. While $\partial \mathcal{T}$ is non-local, the concept of a trapping surface leads to the following quasi-local analogue of a future event horizon \cite{hayward1994a, hayward1994b} (also known as an apparent horizon \cite{kriele1997} in applications): 

\begin{defn}
A \gls{foth} is a smooth 3D submanifold $\underline{\mathcal{H}}$ of spacetime, foliated by closed 2D submanifolds $\underline{S}$, such that
\begin{itemize}
\item the expansion of one future direction null normal to the foliation, say $\ell^a$, vanishes, $\theta_{(\ell)} =0$; 
\item the expansion of the other future directed null normal, $n^a$, is negative, $\theta_{(n)} < 0$; and 
\item the directional derivative of $\theta_{(n)}$ along $n^a$ is negative, $\mathcal{L}_n \theta_{(\ell)} <0$. 
\end{itemize}
\end{defn}

\noindent The Raychaudhuri equation shows that $\underline{\mathcal{H}}$ is either spacelike or null if the shear $\sigma_{ab}$ of $\ell^a$ and the matter flux $T_{ab} \ell^a \ell^b$ across $\underline{\mathcal{H}}$ vanish. In this case $\underline{\mathcal{H}}$ is a NEH (i.e., a WIH). The FOTH is spacelike in the dynamical region where gravitational radiation and matter fields are pouring into it, and is null when it has reached equilibrium. By relaxing the condition that $\mathcal{L}_n \theta_{(\ell)}$ is negative we have the definition of a dynamical horizon: 
 
\begin{defn}
A smooth, 3D spacelike submanifold (possibly with boundary), $\underline{\mathcal{H}}'$ of spacetime is said to be a \gls{dh} if it can be foliated by closed 2D submanifolds $\underline{S}'$, such that
\begin{itemize}
\item the expansion of one future direction null normal to the foliation, say $\ell^a$, vanishes, $\theta_{(\ell)} =0$; 
\item the expansion of the other future directed null normal, $n^a$, is negative, $\theta_{(n)} < 0$.
\end{itemize}
\end{defn}

Since MTTs depend on the choice of a reference foliation of spacelike hypersurfaces, they are non-unique.
%MTTs are highly non-unique, because they depend on the choice of a reference foliation of spacelike hypersurfaces. 
The non-uniqueness of trapped surfaces is inherited by everything based on them, such as MTTs and including dynamical horizons. To resolve this we could use the well defined event horizon and accept its teleological properties, treat all possible MTTs and dynamical/trapping horizons as equally valid, use some other non-local boundary, or try to define preferred marginally trapped tubes \cite{Senov}.

A dynamical horizon is better suited to analyse dynamical processes involving black holes, such as black hole growth and coalescence. A dynamical horizon is a 3D spacelike hypersurface foliated by marginally trapped 2D compact surfaces, which can transition to an isolated null NEH when the flux of gravitational radiation or matter across it is zero. Fluxes of energy and angular momentum carried by gravitational waves across a dynamical horizon necessarily cause the area of such surfaces to increase with time, and the corresponding change in the horizon cross section area was analysed in \cite{AK1,AK2}. Due to back scattering, the transition to equilibrium takes an infinite time.  However, considering a finite time transition is more instructive for it involves a smooth matching between dynamical and non-expanding horizons.  As it was illustrated in \cite{AK1,AK2}, angular momentum, energy, area, and surface gravity of the horizons cross sections match smoothly.

The Vaidya solution provides a simple and explicit example of a dynamical horizon \cite{Vaidyaa, Vaidyab, BonnorVaidya}. In addition, for an appropriate mass function the Vaidya solution provides examples of the transition from the dynamical to isolated horizons. The Vaidya solution admits spherically symmetric marginally trapped surfaces. The existence of non-spherically symmetric dynamical horizons which asymptote to the NEH was discussed in \cite{AshtekarKrishnan} where it was shown that if a hypersurface admits a dynamical horizon structure, it is unique. However, if a spacetime has several distinct black holes, it may admit several distinct non-unique  dynamical horizons.  If one considers dynamical horizons which are also FOTHs (spacelike future outer trapping horizons, SFOTHs), then one can show that if two non-intersecting SFOTHs become tangential in a finite time to the same NEH, then they either coincide or one is contained in the other.  However, one cannot rule out the existence of more than one SFOTH which asymptotes to the NEH if they intersect each other repeatedly.

In this paper we will explore the relationship between these surfaces for black holes admitting stationary horizons and NEHs, and for the spherically symmetric dynamical black hole solutions. We will introduce a new surface to study, defined by the requirement that the Ricci and Weyl tensors are more algebraically special on this surface as compared to the rest of the spacetime. This condition will be defined in terms of the vanishing of scalar curvature invariant, which implies that these surfaces are foliation independent. In section 2 we review the discriminant scalar polynomial invariants, and show how they can be used to determine when a spacetime becomes algebraically special. In section 3 we discuss the event horizon for stationary black holes which is a Killing horizon and are detectable by scalar curvature invariants; we posit that these invariants are related to the discriminant invariants. In section 4 we show that other horizons beyond Killing horizons are detectable by invariants, namely the NEHs, WIHs and the dynamical horizon of a spherically symmetric metric. We also discuss how the dynamical horizons of less idealized black hole solutions could be detected using invariants. Motivated by these results we introduce the geometric horizon detection conjectures in section 5. In section 6 we summarize the results and discuss their applications.

There are five appendices. In Appendix A we provide the Kerr-Newman-\glsdisp{nut}{NUT}-(Anti)-de Sitter metric as an example to show that a frame exists for which the curvature tensor and its covariant derivatives becomes algebraically special on the event horizons of this metric. In Appendix B we compare the Page-Shoom invariants and the discriminant scalar polynomial curvature invariants for the Kerr-Newman-NUT-(Anti)-de Sitter metric and show that the invariants share a common factor. In Appendix C we review the geometric identities for the contractions of the Riemann tensor and its covariant derivatives in order to determine a minimal basis for the set of polynomials formed from them and possibly simplify the discriminant invariants. In Appendix D we present the necessary type {\bf II/D} conditions for the Weyl tensor using discriminant invariants. In Appendix E we summarize the abbreviations frequently used in this paper. 
\newpage

\section{Discriminating Scalar Polynomial Curvature Invariants}

The introduction of {\em alignment theory} \cite{classa,classb,classc} allows for the algebraic classification of any tensor in a Lorentzian spacetime of arbitrary dimensions using \gls{bw}. The dimension-independent theory of alignment can be applied to the tensor classification problem for the Weyl tensor in higher dimensions \cite{classa,classb,classc}, and to the classification of second-order symmetric tensors, such as the Ricci tensor, and tensors involving covariant derivatives. The Ricci tensor can also be classified according to its eigenvalue structure. In a related way, the classification  of the Weyl tensor can be obtained by introducing {\em bivectors}, where the Weyl bivector operator is defined in a manner consistent with its b.w. decomposition \cite{BIVECTOR}. 

The classifications of the Weyl tensor are distinct in higher dimensions; however, in four dimensions (4D) they yield the Petrov classificaton \cite{kramer}. Using the b.w. decomposition and curvature operators together, the algebraic classification of the Weyl tensor and the Ricci tensor (and their covariant derivatives) in higher  dimensions can be refined by exploiting their eigenbivector and eigenvalue structure. A tensor of a particular special algebraic type will have an associated operator with a restricted eigenvector structure, and this can then be used to determine necessary conditions for the algebraic type. 

If a tensor is of alignment type {\bf II}, there exists a frame where all components with positive b.w. vanish. If a tensor is of alignment type {\bf D} then there exists a frame where all components with non-zero b.w. vanish. Using discriminants, we can completely determine the eigenvalues of the curvature operator (up to degeneracies) yielding, for example, necessary conditions in terms of simple \glspl{spi} for the Weyl and Ricci curvature operators to be of algebraic type {\bf II} or {\bf D} in arbitrary dimensions \cite{CHDG,CH}. Necessary conditions for the covariant derivatives of the Ricci and Weyl tensor to be algebraically special can be found by forming 2 or 4 index tensors from them.

\subsection{Discriminant Analysis}

A {\em SPI of order $k$} is a scalar obtained by the contraction of copies of the Riemann tensor and its covariant derivatives up to the order $k$.  In arbitrary dimensions, requiring that all of the zeroth order polynomial Weyl invariants vanish implies that the Weyl type is {\bf III}, {\bf N}, or {\bf  O} (similarly for the Ricci type). SPIs have been used in the study of $VSI$ and $CSI$ spacetimes, where all of the SPIs vanish or are constant, respectively \cite{CSI4a,CSI4b,CSI4c,Higher,CFH}. In \cite{CSI4a} it was proven that a 4D Lorentzian spacetime metric is either \emph{$\mathcal{I}$-non-degenerate}, and hence completely locally characterized by its SPIs, or it is either locally homogeneous or \emph{degenerate Kundt} \cite{Kundt}.

For any tensor of b.w. type {\bf II} (or {\bf D}) the eigenvalues of the corresponding operator need to be of a special form; the resulting invariants for a tensor of type {\bf II} are the same as that for a type {\bf D} tensor. For the ensuing discussion we will assume the tensor is of type {\bf II}. If the Ricci tensor is to be of type {\bf II}, it is of Segre type $\{(1,1)11...1\}$, or simpler. Therefore, the \emph{Ricci operator has at least one eigenvalue of (at least) multiplicity 2. Furthermore, all the eigenvalues are real.} In $D$ dimensions, we may consider the Weyl bivector operator $C_{ab}^{~~cd}$ as the map $${\sf C}:  \Lambda^2 M \to \Lambda^2 M$$
\noindent and examine its eigenvalues. If the Weyl tensor is of type {\bf D}, then the \emph{operator ${\sf C}$ has at least $(D-2)$ eigenvalues of (at least) multiplicity 2}  \cite{BIVECTOR}. In this manner, the algebraic types are connected to the eigenvalue structure and allow for the construction of the necessary discriminants.

In $D$ dimensions, the Ricci and Weyl type {\bf II}/{\bf D} necessary conditions are ($m=D(D-1)/2)$:
\beq \text{Ricci:} \quad && {}{^D}D_D=0, \\
\text{Weyl:}\quad && {}^mD_m= {}^{m}D_{m-1}=...={}^mD_{m-D+2}=0, \eeq
\noindent which are discriminants defined in terms of SPIs\cite{CHDG,CH}. Note that the Ricci syzygy, as a polynomial in terms of the Ricci tensor components, is of order $D(D-1)$, while the highest Weyl syzygy is of order $D(D^2-1)(D-2)/4$. We are interested in 4D, and so the Ricci and Weyl syzygies are of order 12 and 30, respectively. These conditions are necessary conditions, but are not sufficient, since the characteristic equation for different algebraic types may be identical and consequently the SPIs are also identical. For example, a five dimensional (5D) spacetime which has a Weyl tensor with $SO(2)$ isotropy fulfills the type {\bf II} or {\bf D} necessary conditions \cite{CHDG,CH}.

\subsubsection{Ricci Type {\bf  II}/{\bf D} in 4D}

To determine the type {\bf II/D} conditons for the Ricci tensor, we consider a general 2-index tensor, which is assumed to be symmetric and trace-free (${\it S_1}=0$) in 4D. The discriminant ${^4}D_4$ is given by:
\beq
{}^4D_4&=& \frac 1{8}\,{{ S_2}}^{6}-{\frac {5}{4}}\,{{ S_2}}^{4}{ S_4}-{\frac {17}{18}}\,{{ S_3}}^{2}{{ S_2}}^{3}\nonumber \\
&&+4\,{{ S_2}}^{2}{{ S_4}}^{2}+{2}\,{{ S_3}}^{2}{ S_2}\,{ S_4}-\frac 13\,{{ S_3}}^{4}-4\,{{ S_4}}^{3},\eeq
where $S_{ab}$ is the trace-free symmetric Ricci tensor $R_{ab} - \frac{1}{4} R g_{ab}$, and $S_i$ is the trace of the $i^{th}$ power of this tensor. 

This 12th order SPI can be written in a shorter form using:
\beq s_{2} &=&  -\frac12 S^a_{~b} S^b_{~a} = -\frac12 S_2, \nonumber \\
s_{3} &=& -\frac13 S^a_{~b} s^b_{~c} s^c_{~a} = -\frac13 S_3, \nonumber \\
s_{4} &=& \frac18 (S^a_{~b} S^b_{~a})^2 - \frac14 S^a_{~b} S^b_{~c} S^c_{~d} S^d_{~a} = \frac18 S_2^2 - \frac14 S_4. \label{riccidef}  \eeq
\noindent The condition ${^4}D_4=0$ for the 4D symmetric trace-free Ricci  tensor to necessarily be of Ricci type {\bf  II}/{\bf D} in 4D is then \cite{CHDG,CH}:

\begin{equation}
%\label{mathcalD}
\mathcal{D} \equiv {^4}D_4 = -s_3^2(4 s_2^3 - 144 s_2 s_4 + 27s_3^2) + s_4(16 s_2^4 - 128 s_4 s_2^2  + 256 s_4^2) = 0.
\label{rictypeii2}
\end{equation}

\subsubsection{Weyl Type {\bf  II}/{\bf D} in 4D}

Necessary and sufficient real conditions for the Weyl tensor to be of type {\bf  II}/{\bf D} are given by the vanishing of the two SPIs \cite{CHDG,CH}:
\begin{equation}
\mathcal{W}_1 \equiv -11 W_{2}^3 + 33 W_2 W_4 - 18 W_6 , \label{weyl1}
\end{equation}
\begin{equation}
\mathcal{W}_2 \equiv (W_{2}^2 - 2 W_4)(W_{2}^2 + W_4)^2 + 18 W_3^2(6 W_6 - 2 W_{3}^2 -
9 W_{2} W_4 + 3 W_{2}^3) , \label{weyl2}
\end{equation}
where
\begin{eqnarray}
W_2 &=& \frac{1}{8}C_{abcd}C^{abcd},\\ \nonumber
W_3  &=& \frac{1}{16}C_{abcd}C^{cd}_{~~ pq}C^{pqab},\\ \nonumber
W_4 &=& \frac{1}{32}C_{abcd}C^{cd}_{~~pq}C^{pq}_{~~r s}C^{rsab}, \\ \nonumber
W_6 &=& \frac{1}{128}C_{abcd}C^{cd}_{~~pq}C^{pq}_{~~r s}C^{rs}_{~~tu}C^{tu}_{~~vw}C^{vwab}
.
\label{weyldef}
\end{eqnarray}
These two conditions  are equivalent to the real and imaginary parts of the complex syzygy $I^3-27J^2=0$ in terms of the complex Weyl tensor in the \gls{np} formalism \cite{kramer}. Computationally it might be useful to eliminate $W_6$ from  \eqref{weyl1} and \eqref{weyl2} in order to obtain a single necessary condition. We can apply this result to any Weyl candidate (a 4-index tensor with the same symmetries as Weyl tensor) 
\[ C^{c}_{~acb}=0, \quad C_{a(bcd)}=0.\] 

\noindent Alternatively, we can use $I_1 = C_{abcd}C^{abcd}$  to construct the the trace-free operator: $$T_a^{~e} = C_{abcd}C^{ebcd} - \frac{I_1}{4} \delta_a^{~e}$$ with the invariants:

\begin{eqnarray}
\tilde{W}_4&\equiv&
C^{abcd}C_{ebcd}
C^{eb_1c_1d_1}C_{ab_1c_1d_1}
\\ \nonumber
\tilde{W}_6&\equiv&
C^{abcd}C_{a_1bcd}
C^{a_1 b_1c_1d_1}C_{a_2b_1c_1d_1}
C^{a_2b_2c_2d_2}C_{ab_2c_2d_2}
\\ \nonumber
\tilde{W}_8&\equiv&
C^{abcd}C_{a_1bcd}
C^{a_1 b_1c_1d_1}C_{a_2b_1c_1d_1}
C^{a_2b_2c_2d_2}C_{a_3b_2c_2d_2}
C^{a_3b_3c_3d_3}C_{ab_3c_3d_3}.\
 \end{eqnarray}
The discriminant analysis also gives the coefficients of the characteristic equation as: 
\begin{eqnarray}
w_{2} &=& -\frac12 \tilde{W}_4 + \frac18 I_1^2 ,\\ \nonumber
w_{3} &=& -\frac16 \tilde{W}_6 + \frac14 I_1 \tilde{W}_4-\frac{1}{24}I_1^3 ,\\ \nonumber
w_{4} &=& \frac18 \tilde{W}_4^2 + \frac{1}{32}I_1^2 \tilde{W}_4 + \frac{5}{256}I_1^4-\frac14 \tilde{W}_8+\frac{1}{4}I_1 \tilde{W}_6.\label{NWeylRicci}
\end{eqnarray}
\noindent Therefore, the necessary condition for the operator  $T_a^{~e}$ to be type {\bf II/D} is similar to equation \eqref{rictypeii2}: 
\begin{equation}
{^4}D_4 \equiv -{^2}w_3^2(4 {^2}w_2^3 - 144 {^2}w_2 {^2}w_4 + 27{^2}w_3^2) + {^2}w_4(16 {^2}w_2 - 128 {^2}w_4 {^2}w_2^2  + 256 {^2}w_4^2) = 0
\label{rictypeii21}
\end{equation}

The discriminant analysis provides syzygies expressed in terms of the SPIs by treating the Weyl tensor as a trace-free operator acting on the six-dimensional vector space of bivectors. The type {\bf II/D} condition is ${^6}D_5 = {^6}D_6 = 0$; however, these conditions are very large (see Appendix D). 

As the necessary and sufficient conditions  \eqref{weyl1} and \eqref{weyl2} are of lower order than the corresponding discriminant SPI for the Weyl tensor, it is possible that this discriminant SPI can be factored.

\subsubsection{The Riemann tensor and other tensors in 4D}

If a spacetime is of Riemann type {\bf II}, the Weyl type {\bf II} and Ricci type {\bf II} necessary conditions hold, and there are additional alignment conditions (e.g., $C_{abcd}R^{bd}, C_{abcd}R^{be}R_{e}^{~d}$ are of type {\bf II}). Applying the necessary conditions to the full Riemann tensor (to be of type {\bf II}/{\bf D}), implies that both the Weyl and Ricci tensor are of type {\bf II}/{\bf D} and aligned. We note that we will also obtain syzygies for mixed tensors of the form:  \[
L_{ab}=C_{acbd}R^{cd}, \quad
M_{ab}=C_{acbd}R^{c}_{~e}R^{ed}, \quad
N_{ab}=C_{cafg}C^{fg}_{~~db}R^{cd},\]
\noindent to be of type {\bf II}/{\bf D}. The type {\bf II}/{\bf D} condition implies that we have the syzygy ${}^4D_4=0$ for \emph{all of} the trace-free tensors arising from ${\sf L}=(L^a_{~b}) $, ${\sf M}=(M^a_{~b}) $, and ${\sf N}=(N^a_{~b})$.

\subsection{Examples}
To illustrate the applicability of the discriminant SPIs we present four examples. 

\paragraph{An arbitrary 5D Spacetime:} For the trace-free Ricci tensor, we note that type {\bf D} has to be of Segre type $\{(1,1)111\}$ or simpler, implying that 2 eigenvalues are equal, while the remaining eigenvalue has to be real. The vanishing of ${^5}D_5$ is a necessary condition for the trace-free tensor $S$ to be of type {\bf II} (or {\bf D}) in 5D. Thus, the 20th order discriminant $\mathcal{D} \equiv {}^5D_5$ is the related SPI. 

For the ten-dimensional Weyl tensor, the type {\bf II}  bivector operator ${\sf C}$ has 3 eigenvalues of minimum multiplicity 2, and the necessary condition for the Weyl tensor to be of type {\bf II} (or {\bf D}) in 5D is the vanishing of the SPIs ${}_W^{10}D_{10}= {}_W^{10}D_{9}= {}_W^{10}D_{8}=0.$ These are discriminants of order 90, 72 and 56, respectively. Additional necessary conditions can also be found using combinations of the Weyl tensor; for example, the operator $T^{a}_{~b}=C^{acde}C_{bcde}$. This gives again ${}_T^5{D}_5=0$ (${}_T^5{D}_4\geq 0$), which is a 20th order syzygy (in the square of the Weyl tensor).

\paragraph{5D Schwarzschild spacetime:}  For the Weyl operator ${\sf C}$ we get \[ {}^{10}D_{10}=
{}^{10}D_{9}=\dots={}^{10}D_4=0, \quad {}^{10}D_3>0, \quad \quad {}^{10}D_2>0.\] This implies that
the Weyl operator has 3 distinct real eigenvalues 
and this spacetime is of type {\bf D} \cite{BIVECTOR}.

\paragraph{5D space with complex hyperbolic sections.}
Let us consider  \cite{BIVECTOR}:
\beq
\d s^2= -\d t^2+a(t)^2\Big{[}e^{-2w}\left(\d x+\tfrac 12(y\d z-z\d y)\right)^2\qquad &&\nonumber \\
 +e^{-w}\left(\d y^2+\d z^2\right)+\d w^2\Big{]}.&&
\label{HC2}\eeq
For the Weyl operator ${\sf C}$ we get
\[ {}^{10}D_{10}= {}^{10}D_{9}=\dots={}^{10}D_4=0, \quad {}^{10}D_3>0, \quad \quad {}^{10}D_2>0.\]
This implies that  the Weyl operator has 3 distinct real eigenvalues. However, the following coefficients of the characteristic equation vanish:
$a_{10}=a_9=a_8=...=a_4=0,$ signaling that there is a zero-eigenvalue of multiplicity 7. Thus, this spacetime is not of type {\bf II}/{\bf D}. In fact, it is $\mathcal{I}$-non-degenerate which can be shown by computing the operator $T^{a}_{~b}=C^{acde}C_{bcde}$ which is of ``Segre'' type $\{1,(1111)\}$.

\paragraph{The 5D rotating black ring.}
The 5D rotating black ring \cite{RBRa,RBRb} is generally of type ${\bf I_i}$, but in certain regions and for particular values of the parameters $\lambda$ and $\mu$ it can also be of type ${\bf II}$ or ${\bf D}$ (the case $\lambda=1$ corresponds to the type ${\bf D}$ Myers-Perry metric) \cite{AM2016, CKAHD}. The trace-free and symmetric part of the operator $T^{a}_{~b}=C^{acde}C_{bcde}$ gives a discriminant which leads to a necessary condition on the algebraic type of the Weyl tensor in the region of Lorentzian signature for the fixed `target' point locally defined by $x=0$, $y=2$ \cite{CHDG,CH}:
\beq 
{}_T^5D_5=\frac{\lambda^{12}(\lambda-\mu)^{12}(2\mu-1)^2(1-\lambda)^4(1+\lambda)^4}{(1-2\lambda)^{113}}F(\mu,\lambda).    
\eeq
where $F(\mu,\lambda)$ is a  polynomial which is generally not zero. For this particular choice of target point, the horizon is located there if $\mu=1/2$, and we see that ${}_T^5D_5=0$ (with ${}_T^5D_4>0$), which signals that the spacetime is of Weyl type {\bf II} on the horizon. Computationally, it is simpler to work with ${\bf T}$ and the 40th order SPI ${}_T^5D_5$ than the related SPI for the Weyl tensor as an operator.

\subsection{Differential invariants}

To determine whether the covariant derivatives of the Ricci tensor $R_{ab;cd...}$,  are also of type {\bf II} or  {\bf D}, we could study the eigenvalue structure of the operators constructed from the tensor $R_{ab;cd...}$ and apply the type {\bf II}/{\bf D} necessary conditions. For example, considering the trace-free parts of the tensors $T_{ab} = R_{ac:d}R_b^{~c;d}, R_{;ab}, \Box R_{ab}, \dots$, we obtain necessary conditions of the form of equation (\ref{rictypeii2}) but with the $s_i \equiv \Tr({\sf T}^i), i=2,3,4$. This can be repeated for the Weyl tensor and in higher dimensions \cite{invhigher}.

For example, we can construct the second order symmetric and trace-free operator, $T^{e}_{~f}$ , for the covariant derivative of the Weyl tensor, $C_{abcd;e}$ defined by:
\[ T^{e}_{~f} \equiv C^{abcd;e}C_{abcd;f}
- \frac{1}{4} \delta^{e}_{~f} C^{abcd;e'}C_{abcd;e'}\]

\noindent The resulting differential invariants may be simplified using the FKWC bases for the Riemann SPIs to eliminate Riemannian SPIs that can be expressed in terms of the bases. Additionally one can use geometric identities and conserved tensor quantities to induce further simplification (see Appendix C).%section \ref{geomid} ).

\paragraph{Example.} We consider the operator $T^{\mu}_{~\nu}$ defined above for the 4D type {\bf D} Kerr metric. The type {\bf D}/{\bf II} necessary condition is then the vanishing of:
\beq \label{kerrdiff}
{}_T^4D_4=\frac{m^{24}a^4 G^2_-G^2_+(r^2+a^2-2mr)^2(r^2+a^2\cos^2\theta-2mr)^2 \sin^4\theta}{(r^2+a^2\cos^2\theta)^{92} }f_1^2f_2, \nonumber 
\eeq
where  $G_{\pm} \equiv r^4\pm 4ar^3\cos\theta-6a^2r^2\cos^2\theta\mp 4a^3r\cos^3\theta+a^4\cos^4\theta$, and $f_1=f_1(a,m,r,\cos\theta)$ and $f_2=f_2(a,m,r,\cos\theta)$ are polynomials. With the exception of the horizon, the ergosphere, and some other special points, this invariant will be non-zero and so $C_{abcd;e}$ cannot be of type {\bf D}/{\bf II} (generically) outside the horizon. This is confirmed explicitly using the Cartan algorithm to construct the appropriate frame in Appendix A. %section \ref{4Dexa}.

\subsubsection{Necessary 4D conditions for the covariant derivative of the Weyl tensor to be of type {\bf  II}/{\bf D}}

To determine the algebraic type of the covariant derivative of the Weyl tensor, $C_{abcd;\mu}$, we consider a second order symmetric and trace-free operator $T^{a}_{~b}$ to  obtain the necessary type {\bf  II}/{\bf D} condition (\ref{rictypeii2}) of the form $\mathcal{D} \equiv {^4}D_4=0$.

Let us consider two possible combinations, and from both we can derive necessary conditions. The first is the trace-free symmetric tensor ${^1}S^{a}_{~b}={^1}T^{a}_{~b}$ defined above:
\begin{equation}
{^1}T^{a}_{~b} \equiv C^{efcd;a}C_{efcd;b}
- \frac{1}{4} \delta^{a}_{~b} {^1}I_2, \label{T1a}
\end{equation}
where
\begin{equation}
{^1}I_2\equiv
C^{abcd;e}C_{abcd;e}, \label{T1b}
\end{equation}
and we define
\small
\begin{eqnarray}
{^1}I_4&\equiv&
C^{abcd;e}C_{abcd;e_1}
C^{a_1b_1c_1d_1;e_1}C_{a_1b_1c_1d_1;e}
\\ \nonumber
{^1}I_6&\equiv&
C^{abcd;e}C_{abcd;e_1}
C^{a_1b_1c_1d_1;e_1}C_{a_1b_1c_1d_1;e_2}
C^{a_2b_2c_2d_2;e_2}C_{a_2b_2c_2d_2;e}
\\ \nonumber
{^1}I_8&\equiv&
C^{abcd;e}C_{abcd;e_1}
C^{a_1b_1c_1d_1;e_1}C_{a_1b_1c_1d_1;e_2}
C^{a_2b_2c_2d_2;e_2}C_{a_2b_2c_2d_2;e_3}
C^{a_3b_3c_3d_3;e_3}C_{a_3b_3c_3d_3;e}.\
 \end{eqnarray}
 \normalsize
\noindent Computing the coefficients of the characteristic equation in terms of these:

\begin{eqnarray}
{^1}s_{2} &=& -\frac12 {^1}I_4 + \frac18 {^1}I_2^2 ,\\ \nonumber
{^1}s_{3} &=& -\frac13 {^1}I_6 + \frac14 {^1}I_2 {^1}I_4-\frac{1}{24}{^1}I_2^3 ,\\ \nonumber
{^1}s_{4} &=& \frac18 {^1}I_4^2 - \frac{5}{32}{^1}I_2^2 {^1}I_4 + \frac{5}{256}{^1}I_2^4-\frac14 {^1}I_8+\frac{1}{4}{^1}I_2 {^1}I_6         .\label{riccidef12}
\end{eqnarray}

\noindent The necessary condition for this operator to be of type {\bf II/D} is equivalent in form to the condition given in equation \eqref{rictypeii2}:

\begin{equation}
{^1} D = {^1}X \equiv -{^1}s_3^2(4 {^1}s_2^3 - 144 {^1}s_2 {^1}s_4 + 27{^1}s_3^2) + {^1}s_4(16 {^1}s_2^4 - 128 {^1}s_4 {^1}s_2^2  + 256 {^1}s_4^2) = 0
\label{rictypeii212}
\end{equation}
%{^1}\mathcal{D}={^1}s_{3}^2(4{^1}s_{2}^3-6{^1}s_{2}{^1}s_{4}+{^1}s_{3}^2)-{^1}s_{4}^2(3{^1}s_{2}^2-4{^1}s_{4})=X=0,

\noindent Expanding this expression, we obtain explicitly:
\beq && {^1}X =  \frac83 \,{  {^1}I_2}\,{  {^1}I_6}\,{{  {^1}I_4}}^{4}-{\frac {25}{2}}\,{{  {^1}I_6}}^{2
}{  {^1}I_8}\,{{  {^1}I_2}}^{2}-4\,{{  {^1}I_8}}^{3}-\frac13\,{{  {^1}I_6}}^{4}+1/8\,{
{  {^1}I_4}}^{6}+{\frac {1}{576}}\,{{  {^1}I_2}}^{12}\nonumber \\
&&-11\,{  {^1}I_8}\,{  {^1}I_2} 
\,{  {^1}I_6}\,{{  {^1}I_4}}^{2} -{\frac {15}{4}}\,{  {^1}I_8}\,{{  {^1}I_2}}^{3}{
  {^1}I_6}\,{  {^1}I_4}-{\frac {27}{16}}\,{  {^1}I_8}\,{{  {^1}I_4}}^{2}{{  {^1}I_2}}
^{4}+12\,{{  {^1}I_8}}^{2}{  {^1}I_2}\,{  {^1}I_6} \nonumber \\
&&-{\frac {7}{32}}\,{  {^1}I_8}\,{
  {^1}I_4}\,{{  {^1}I_2}}^{6} -\frac74 \,{  {^1}I_8}\,{{  {^1}I_2}}^{5}{  {^1}I_6} +{\frac {
73}{48}}\,{{  {^1}I_2}}^{5}{  {^1}I_6}\,{{  {^1}I_4}}^{2}+\frac{3}{16} \,{{  {^1}I_2}}^{7}{
  {^1}I_6}\,{  {^1}I_4}+{\frac {73}{72}}\,{{  {^1}I_2}}^{3}{  {^1}I_6}\,{{  {^1}I_4}}
^{3}\nonumber \\
&&+{\frac {35}{6}}\,{{  {^1}I_6}}^{2}{{  {^1}I_2}}^{2}{{  {^1}I_4}}^{2}+{
\frac {37}{24}}\,{{  {^1}I_6}}^{2}{  {^1}I_4}\,{{  {^1}I_2}}^{4}+2\,{{  {^1}I_6}}^
{2}{  {^1}I_8}\,{  {^1}I_4}+\frac52 \,{{  {^1}I_8}}^{2}{{  {^1}I_2}}^{2}{  {^1}I_4}-{{  
{^1}I_6}}^{3}{  {^1}I_2}\,{  {^1}I_4} \nonumber \\
&& -\frac58 \,{  {^1}I_8}\,{{  {^1}I_2}}^{2}{{  {^1}I_4}}^{3}
+\frac{13}{3} \,{{  {^1}I_6}}^{3}{{  {^1}I_2}}^{3} -{\frac {17}{18}}\,{{  {^1}I_6}}^{2}{{
  {^1}I_4}}^{3} +{\frac {17}{18}}\,{{  {^1}I_6}}^{2}{{  {^1}I_2}}^{6}+4\,{{  
{^1}I_8}}^{2}{{  {^1}I_4}}^{2}+{\frac {13}{16}}\,{{  {^1}I_8}}^{2}{{  {^1}I_2}}^{4} \nonumber \\
&&-
\frac14 \,{{  {^1}I_2}}^{2}{{  {^1}I_4}}^{5}+{\frac {61}{192}}\,{{  {^1}I_4}}^{4}{{
  {^1}I_2}}^{4}+{\frac {95}{576}}\,{{  {^1}I_4}}^{3}{{  {^1}I_2}}^{6} +{\frac {
55}{768}}\,{{  {^1}I_4}}^{2}{{  {^1}I_2}}^{8}+{\frac {1}{192}}\,{  {^1}I_4}\,{{
  {^1}I_2}}^{10}\nonumber \\
&&-\frac54 \,{  {^1}I_8}\,{{  {^1}I_4}}^{4}-\frac{1}{16} \,{  {^1}I_8}\,{{  {^1}I_2}}
^{8}+{\frac {5}{72}}\,{{  {^1}I_2}}^{9}{  {^1}I_6}.
 \label{BigX1}  \eeq

Alternatively, we can construct the trace-free symmetric tensor
${^2}S^{a}_{~b} =
{^2}T^{a}_{~b}$ where
\begin{equation}
{^2}T^{a}_{~b} \equiv C^{afcd;e}C_{bfcd;e}
- \frac{1}{4} \delta^{a}_{~b} {^2}I_2,
\end{equation}
 where
\begin{equation}
{^2} I_1 \equiv {^1}I_2\equiv
C^{abcd;e}C_{abcd;e},
\end{equation}
and
\begin{eqnarray}
{^2}I_4&\equiv&
C^{abcd;e}C_{a_1bcd;e}
C^{a_1b_1c_1d_1;e_1}C_{ab_1c_1d_1;e_1}
\\ \nonumber
{^2}I_6&\equiv&
C^{abcd;e}C_{a_1bcd;e}
C^{a_1b_1c_1d_1;e_1}C_{a_2b_1c_1d_1;e_1}
C^{a_2b_2c_2d_2;e_2}C_{ab_2c_2d_2;e_2}
\\ \nonumber
{^2}I_8&\equiv&
C^{abcd;e}C_{a_1bcd;e}
C^{a_1b_1c_1d_1;e_1}C_{a_2b_1c_1d_1;e_1}
C^{a_2b_2c_2d_2;e_2}C_{a_3b_2c_2d_2;e_2}
C^{a_3b_3c_3d_3;e_3}C_{ab_3c_3d_3;e_3}.\
 \end{eqnarray}
\noindent Computing the coefficients of the characteristic equation in terms of these:
\begin{eqnarray}
{^2}s_{2} &=& -\frac12 {^2}I_4 + \frac18 {^2}I_2^2 ,\\ \nonumber
{^2}s_{3} &=& -\frac13 {^2}I_6 + \frac14 {^2}I_2 {^2}I_4-\frac{1}{24}{^2}I_2^3 ,\\ \nonumber
{^2}s_{4} &=& \frac18 {^2}I_4^2 - \frac{5}{32} {^2}I_4 ({^2}I_2^2) + \frac{5}{256}{^2}I_2^4-\frac14 {^2}I_8+\frac{1}{4}{^2}I_2 {^2}I_6         .\label{riccidef1}
\end{eqnarray}

\noindent The necessary conditions will take the form of \eqref{rictypeii2} or \eqref{rictypeii212}; explicitly this will be given by \eqref{BigX1} where the index ${^1}$ is replaced by ${^2}$. 

\subsubsection{Necessary 4D conditions for the second covariant derivative of the Weyl tensor to be of type {\bf  II}/{\bf D}}

The second covariant derivative of the Weyl tensor, $C_{abcd;ef}$, can be studied using a second order symmetric and trace-free operator $T^{a}_{~b}$ to  obtain the necessary type {\bf  II}/{\bf D} conditions of the form $\mathcal{D} \equiv {^4}D_4=0$. The simplest operator to consider is the following: 
\beq  \tilde{T}^{a}_{~b} \equiv  C^{a~~~;cd}_{~cbd}, \nonumber \eeq
\noindent and defining $\tilde{T}_i$ as the trace of the $i^{th}$ power of this tensor, this expression can be written in a shorter form using:
\beq t_{2} &=&  -\frac12 \tilde{T}^a_{~b} \tilde{T}^b_{~a} = -\frac12 \tilde{T}_2, \nonumber \\
t_{3} &=& -\frac13 \tilde{T}^a_{~b} \tilde{T}^b_{~c} \tilde{T}^c_{~a} = -\frac13 \tilde{T}_3, \nonumber \\
t_{4} &=& \frac18 \tilde{T}^a_{~b} \tilde{T}^b_{~a} - \frac14 \tilde{T}^a_{~b} \tilde{T}^b_{~c} \tilde{T}^c_{~d} \tilde{T}^d_{~a} = \frac18 \tilde{T}_2 - \frac14 \tilde{T}_4 \label{riccideft}  \eeq
\noindent The necessary type {\bf II/D} condition, which is of the same form as  \eqref{rictypeii2} and \eqref{rictypeii212}, is then:
\begin{equation}
{^4}D_4 = \tilde{X} \equiv -t_3^2(4 t_2^3 - 144 t_2 t_4 + 27t_3^2) + t_4(16 t_2 - 128 t_4 t_2^2  + 256 t_4^2) = 0
\end{equation}

\section{Scalar invariants in stationary spacetimes and type D conditions}

For a 4D Lorentzian manifold ${\mathcal M}$, if there exists any other geometric (invariantly defined) structure, such as for example an invariantly defined timelike vector field, we can construct other invariants, potentially of lower order. As an illustration, perfect fluid solutions have a timelike Ricci eigenfunction and a stationary spacetime has a timelike Killing vector.

A vector field ${\bf \zeta}$ on  ${\mathcal M}$ is called a {\em Killing vector} if satisfies $\pounds_{{\bf \zeta}} g_{ab}=0$, where $\pounds_{{\bf \zeta}}$ denotes the Lie derivative with respect to ${\bf \zeta}$. This condition is equivalent to $\nabla_{(a}\zeta_{b)} = 0$, which implies $F_{ab}\equiv \nabla_{a}\zeta_{b}$ is a closed (Killing) 2-form, and using the Ricci identity its covariant derivative becomes:
\begin{equation} \nabla_b F_{ac}=-R_{acbd}\zeta^d.  \label{eq:cdfz}
\end{equation}
In general, neither the norm of the Killing vector $\lambda\equiv\zeta_a\zeta^a$   nor the the spacetime dimension and its matter content need be restricted. However, in the particular case of Ricci-flat 4D spacetimes there are additional properties \cite{GPS}.

Employing {\em complex quantities} one can simplify the computations considerably.  The {\em self-dual} Weyl tensor, and the self-dual Killing 2-form, ${\mathcal F}_{ab}\equiv F_{ab}+\mbox{i}F^*_{ab}$, satisfying the algebraic properties displayed in \cite{GPS}, can be defined.
Then the differential conditions satisfied by $F_{ab}$ and $F^*_{ab}$ are: 
\begin{equation}
\nabla_{c}{\mathcal F}_{ab}=-{\mathcal C}_{abcd}\zeta^d\;,\quad \nabla_{[c}{\mathcal
F}_{ab]}=0.  
\label{eq:maxwell} \end{equation} 
Employing the Ernst 1-form $$\sigma_a\equiv{2\mathcal F}_{ab}\zeta^b, \hspace{1cm} (\sigma_a\zeta^a  =0)$$ then in a Ricci-flat spacetime (or more generally when $\zeta_{[c}R_{a]b}\zeta^b=0$ is satisfied)
\begin{equation}
\nabla_{[a}\sigma_{b]}=0\;, 
\end{equation}  
\noindent So that $\sigma_a$ is exact and there exists a local potential $\sigma$ for $\sigma_a$, known as the Ernst potential, which is a complex quantity constructed from the Killing norm and twist as $\sigma\equiv\lambda+2\;\mbox{i}\;\omega$ $\Rightarrow \nabla_a\sigma=\sigma_a$. The scalar $\sigma$ is defined up to the addition of an additive complex constant $\alpha$, giving the gauge freedom $\sigma \longrightarrow \sigma' =\sigma +\alpha$.
The tensor $\mathcal{F}_{ab}$ and scalar $\sigma_a$ are related through the tensor identity 
\begin{equation}
-\lambda\mathcal{F}_{ab}=\zeta_{[a}\sigma_{b]}+\frac{\rm
i}{2}\eta_{abcd}\zeta^c\sigma^d.  \label{eq:decomposeF} 
\end{equation}

At those points of $\mathcal{M}$ where the Ernst 1-form is exact and its potential $\sigma'$ does not vanish, the complex rank-4 {\em{Mars-Simon tensor} } can be defined \cite{MARS-KERR}
\begin{equation} 
\mathcal{S}_{abcd}\equiv \mathcal{C}_{abcd}+\frac{1}{\sigma'}\mathcal{Q}_{abcd}\;, 
\label{def:ms-tensor} \end{equation} 
where
\begin{equation}
\mathcal{Q}_{abcd}\equiv
6\left(\mathcal{F}_{ab}\mathcal{F}_{cd}-\frac{\mathcal{I}_{abcd}}{3}\mathcal{F}\cdot
\mathcal{F}\right), ~~~
{\mathcal{I}}^{ab}{}_{cd}\equiv\frac{1}{4}(\mbox{i}\;\eta^{ab}{}_{cd}+
\delta^a_c\delta^b_d-\delta^a_d\delta^b_c)\;.
\label{eq:kerr-cond-1} 
\end{equation} 
The Mars-Simon tensor is a Weyl candidate, implying that it has the same algebraic properties as the Weyl tensor, and it is self-dual. Its covariant divergence in Ricci-flat spacetimes is a linear expression in the Mars-Simon tensor itself. Note that $\mathcal{S}_{abcd}$ is affected by the residual gauge freedom. The tensor  $\mathcal{Q}_{abcd}$ is also a Weyl candidate. 

It is of interest to determine when a stationary spacetime is "close" or diffeomorphic to the Kerr metric. If a Lorentzian manifold is equivalent to the Kerr solution, then there exists a scaling of the Ernst potential such that  ${\cal S}_{abcd} = 0$. We can also define a tensor $\mathcal{S}_{abc}$ called the {\em spacetime Simon tensor} \cite{MARS-KERR} and for any spacetime with a Killing vector, this tensor can be defined independently of the existence of a potential $\sigma$ for $\sigma_a$. $\mathcal{S}_{abc}$ has the algebraic properties of a Lanczos potential and is totally {\em orthogonal} to the Killing vector $\zeta$ from which it is constructed. The vanishing of $S_{abc}$ implies $C_{abcd}$ and $\mathcal{Q}_{abcd}$ are proportional.

In addition, given a real Weyl candidate $W_{abcd}$, its Bel-Robinson tensor $T_{abcd}\{W\}$ is defined    in terms of the self-dual tensor $\mathcal{W}_{abcd}=W_{abcd}+{\rm i} W^*_{abcd}$ corresponding to the Weyl candidate and its dual \cite{SUPERENERGY}: 
\begin{equation} 
T_{abcd}\{{\mathcal W}\}\equiv
\mathcal{W}_{a\phantom{p}c}^{\phantom{a}p\phantom{d}m} \mathcal{\overline{W}}_{bpdm} =T_{abcd}\{W\} .  
\label{eq:nueva} \end{equation} 
%\begin{equation} T_{abcd}\{W\}\equiv
%W_{a\phantom{p}d}^{\phantom{a}p\phantom{d}m} W_{bp cm}+
%W_{a\phantom{\sigma}c}^{\phantom{a}m\phantom{c}p} W_{bm dp}
%-\frac{1}{8}g_{ab}g_{cd}W_{mnpl} W^{mnpl}.  
%\label{b-r} \end{equation} 
This tensor is the basic {\em superenergy tensor} of the Weyl candidate $W_{abcd}$ in 4D and satisfies additional properties \cite{SUPERENERGY}; applying this to the Weyl tensor yields the Bel-Robinson tensor. The vanishing of the self-dual Weyl candidate ${\mathcal W}_{abcd}$ can be written as a scalar condition.  For any timelike vector ${\bf u}$ we can compute the {\em superenergy density} of $T_{abcd}\{W\}$ (which explicitly depends on the timelike vector ${\bf u}$):
\begin{equation}
U_{{\bf u}}({\mathcal W})\equiv T_{abmn}\{W\}u^{a}u^{b}u^{m}u^{n}
\label{eq:general-superenergy} \end{equation} 
If $W_{abcd}\neq 0$, $U_{{\bf u}}({\mathcal W})$ is a positive quantity, which for each timelike vector acts as a measure of the ``proximity'' to the geometric conditions determined by the tensor condition ${\mathcal W}_{abcd}=0$. 

In those cases where there is a timelike vector  ${\bf \zeta}$ defined invariantly a superenergy density that is invariant can be selected.  For example, this can be implemented for the particular case in which ${\mathcal W}_{abcd}=\mathcal{S}_{abcd}$. The scalar $U_{{\bf \zeta}}(\mathcal{S})$ is positive and it will vanish if and only if $\mathcal{S}_{abcd}=0$, and we can consequently take the quantity $U_{{\bf \zeta}}(\mathcal{S})$ as a local invariant measure of the deviation of the spacetime to the geometric conditions entailed by $\mathcal{S}_{abcd}=0$. This scalar is not a SPI; however, its construction is similar to the invariants discussed in \cite{PageInv} and hence are likely Cartan invariants. Additionally, in \cite{AbdelqaderLake2015} a SPI is given that measures the "Kerrness" for a given Lorentzian manifold; this suggests the existence of a collection of SPIs which could characterize the Kerr metric.

\subsection{Stationary Black Hole Horizon Detection}

The event horizon for a stationary black hole is a null hypersurface that is orthogonal to a Killing vector field that is null on this surface, and hence lies within the hypersurface and is its null generator. For several stationary type {\bf D} solutions, it is known that the SPI $R_{abcd;e}R^{abcd;e}$ detects the event horizon \cite{PRM1993}. However, in the case of the Kerr horizon, it was noted by Skea that this invariant detects the stationary limit, and not the outer horizon itself \cite{SkeaPhd}. A collection of invariants were  examined in \cite{AbdelqaderLake2015}, from which the parameters of the Kerr spacetime (including the detection of the horizons) were determined. These invariants are constructed from SPIs built from the Weyl tensor: 
\small
\beq &Q_1 = \frac{1}{3\sqrt{3}} \frac{(I_1^2 - I_2^2)(I_5-I_6)+4I_1 I_2 I_7}{(I_1^2 + I_2^2)^\frac94},& \nonumber \\~~ &Q_2 = \frac{1}{27}\frac{I_5I_6 - I_7^2}{(I_1^2 + I_2^2)^\frac52},~~Q_3 = \frac{1}{6\sqrt{3}} \frac{I_5 + I_6}{(I_1^2 + I_2^2)^\frac54}, & \label{Qs} \eeq
\normalsize
\noindent in terms of  from the following SPIs:  

%\beq Q_1 &=& \frac{1}{3\sqrt{3}} \frac{(I_1^2 - I_2^2)(I_5 - I_6)+4I_1 I_2 I_7}{(I_1^2+I_2^2)^\frac{9}{4}} \label{Q1} \\
%Q_2 &=& \frac{1}{27} \frac{I_5 I_6 - I_7^2}{(I_1^2 + I_2^2)^\frac{5}{2}} \label{Q2} \\
%Q_3 &=& \frac{1}{6\sqrt{3}} \frac{I_5 + I_6}{I_1^2 + I_2^2)^\frac{5}{4}} \label{Q3} \eeq
%
%\noindent where 
\beq I_1 &=& C_{abcd} C^{abcd} \label{I1} \\
I_2 &=& C^{\star}_{abcd} C^{abcd} \label{I2} \\
I_3 &=& \nabla_e C_{abcd} \nabla^e C^{abcd} \label{I3} \\
I_4 &=& \nabla_e C_{abcd} \nabla^e C^{\star abcd } \label{I4} \\ 
I_5 &=& k_e k^e,~~ k_e = -\nabla_e I_1 \label{I5} \\
I_6 &=& l_e l^e,~~ l_e - \nabla_e I_2 \label{I6} \\
I_7 &=& k_e l^e, \label{I7} \eeq

\noindent where $C^\star_{abcd}$ is the dual of the Weyl tensor. 

%Note that $(I_1^2+I_2^2)$ is a
%positive-definite factor, just used to make $Q_1,$ $Q_2,$ and $Q_3$ dimensionless. 

The parameters are found by using the dimensionless invariants $Q_1$, $Q_2$ and $Q_3$ to locate the horizon and ergosurface in an algebraic manner. The local method does not require knowledge the location of the black hole or its event horizon.  By knowing the forms of the SPIs $I_1,...,I_7$, the mass and angular momentum can be expressed as functions in terms of these invariants.  In addition, a new syzygy for the Kerr spacetime was presented, which allows for the definition of an invariant measure of the "Kerrness" of a spacetime locally: 
\beq I_6 - I_5 + \frac{12}{5} (I_1 I_3 - I_2 I_4) = 0. \label{syz1} \eeq 
As an extension of \cite{AbdelqaderLake2015}, the syzygy \eqref{syz1} was shown to arise from the real-part of a complex syzygy and that by combining the real and imaginary part of this complex syzygy it is possible to write $Q_2$ as the norm of the wedge product of the gradients of $I_1$ and $I_2$, the Kretschmann and Chern-Pontryagin invariants \cite{PageShoom2015}. From this result the authors introduced a general approach to determine the location of the event horizon and ergosurface for any stationary horizon of a black hole. %and the approximate location for any nearly stationary horizon. 

This result relies on the fact that the squared norm of the wedge product of $n$ gradients of functionally
independent local smooth curvature invariants will always vanish on the horizon of any stationary black hole, where
$n$ is the local cohomogeneity of the metric which is the codimension of the maximal dimensional orbits of the
isometry group of the local metric. This leads to the theorem \cite{PageShoom2015}:

\begin{thm} \label{PSthrm}
For a spacetime of local cohomogeneity $n$ that contains a stationary horizon, and which has $n$ SPIs $S^{(i)}$ whose gradients are well-defined there, the $n$-form wedge product 
\beq {\bf{W}} = dS^{(1)} \wedge ... \wedge dS^{(n)} \nonumber \eeq
\noindent has zero squared norm on the horizon: 

\beq ||{\bf{W}}||^2 = \frac{1}{n!} \delta^{\alpha_1,...,\alpha_n}_{\beta_1,...,\beta_n} g^{\beta_1 \gamma_1} ... g^{\beta_n \gamma_n} \times S^{(1)}_{;\alpha_1}...S^{(n)}_{;\alpha_n} S^{(1)}_{;\gamma_1}...S^{(n)}_{;\gamma_n} = 0. \nonumber \eeq
\noindent The permutation tensor $\delta^{\alpha_1,...,\alpha_n}_{\beta_1,...,\beta_n}$, is $+1$ or $-1$ if $\alpha_1,...,\alpha_n$ is an even or odd permutation of $\beta_1,...,\beta_n$ respectively, and is zero otherwise. 
\end{thm} 

%xxyyzz Compare with Alan's edits 

There is then the problem of the existence of the $n$ functionally independent invariants involved in Theorem \ref{PSthrm}. However, Theorem \ref{PSthrm} can be generalized to the set of Cartan invariants in Theorem \ref{DPSthrm} below, and this problem is thus automatically solved since the number of functionally independent Cartan invariants, $t_p$, at the end of the algorithm is related to the dimension of the cohomogeneity \cite{GANG}.

\begin{thm} \label{DPSthrm}
For a spacetime of local cohomogeneity $n$ that contains a stationary horizon, and which has $n$ Cartan invariants $C^{(i)}$ whose gradients are well-defined there, the $n$-form wedge product 
\beq {\bf{W}} = dC^{(1)} \wedge ... \wedge dC^{(n)} \nonumber \eeq
\noindent has zero squared norm on the horizon: 

\beq ||{\bf{W}}||^2 = \frac{1}{n!} \delta^{\alpha_1,...,\alpha_n}_{\beta_1,...,\beta_n} g^{\beta_1 \gamma_1} ... g^{\beta_n \gamma_n} \times C^{(1)}_{;\alpha_1}...C^{(n)}_{;\alpha_n} C^{(1)}_{;\gamma_1}...C^{(n)}_{;\gamma_n} = 0. \nonumber \eeq

\end{thm} 

\noindent Using the Cartan equivalence algorithm, the Cartan invariants (those arising from the covariant derivatives of the Riemann tensor) can be used to produce new invariants that detect the stationary horizons. These invariants will be much simpler to compute than the SPIs \cite{GANG}.

In the following section we will show that that the curvature tensor and its covariant derivatives are algebraically special on an isolated horizon, and so the associated discriminant SPI for any operator constructed from the curvature tensor and its covariant derivatives (see below) will vanish on this surface. As a Killing horizon is a special case of an isolated horizon, the Page-Shoom invariants and the discriminants for the covariant derivative of the Weyl tensor, $C_{abcd;e}$, share a common zero in stationary spacetimes.  This suggests that the Page-Shoom invariants might  indicate where the curvature tensor becomes type {\bf II/D} for stationary spacetimes, and hence provide a computationally simpler alternative to the discriminant SPIs; this is illustrated for the Kerr-Newman-NUT-(Anti)-de Sitter solution in Appendix B.  We conjecture that it is possible to use stationarity to simplify syzygies and relate the discriminant invariants for the Ricci and Weyl tensors and their covariant derivatives to the Page-Shoom invariants. For example, if there exists an invariantly defined timelike Killing vector field we can construct lower order invariants which are related to the Riemann tensor.  %Due to the extensive size of these invariants relative to a coordinate system, we will use the NP formalism to show this property relative to the frame given in \cite{GANG}.   

\newpage

\section{Horizon Examples}

In this section we provide examples of non-stationary spacetimes in which curvature invariants are able to identify a geometrically preferred quasi-local surface, which will be called a {\it geometric horizon}. We will show that NEHs (and hence WIHs) can be detected by the vanishing of SPIs, which provides an alternative to the Page-Shoom invariants as detectors of stationary horizons. Additionally we will show that the dynamical horizons in the class of imploding spherically symmetric metrics can be detected by SPIs, and thus are geometric horizons. That is, we will construct a SPI which vanishes on the horizon, implying that this surface is foliation independent. 

In the case of NEHs (and hence WIHs), we will identify the coframe in which the Riemann tensor and its covariant derivatives are algebraically special on the horizon, which implies that a set of discriminant SPIs will vanish on the NEH. To move beyond equilibrium states of dynamical black holes, we briefly discuss current approaches to identifying quasi-local surfaces bounding dynamical black holes. In the case of imploding spherically symmetric metrics, we will show that the dynamical horizon is detected by SPIs. Therefore, the dynamical horizons of the Vaidya and \gls{ltb} solutions are detectable and foliation independent. This suggests that in the case of dynamical horizons, the geometric horizons will be a set of preferred quasi-local surfaces defined by the vanishing of appropriate SPIs. We shall also discuss identifying geometric horizons using Cartan invariants.

\subsection{Weakly Isolated Horizons}

It is known that the Riemann tensor is algebraically special on a WIH \cite{Lewandowski}. However, the covariant derivatives of the Riemann tensor are also algebraically special. That is, there is a frame in which the Riemann tensor and all higher derivatives are of type {\bf II/D} on the horizon and this will be reflected in the vanishing of the discriminant SPIs. We expect this result will be helpful in the case of spacetimes admitting a WIH and for which the algebraic Riemann tensor is itself of type {\bf II/D}.

\begin{thm} \label{thm:WIH}

On any weakly-isolated horizon (WIH) the Riemann tensor and its covariant derivatives are all of type {\bf II}.
\end{thm}

\begin{proof}

To show that the covariant derivative of the Riemann tensor is of type {\bf II} on the horizon, we must show that  $R_{ab;c}$ and $C^a_{~bcd;e}$ are of type {\bf II} on the NEH $\mathcal{H}$\footnote{We will use different notation from \cite{Ashtekar} for the NEH; instead of $\Delta$ we will denote it by $\mathcal{H}$.}. The type {\bf II} condition for a tensor requires the existence of some frame where all positive b.w. components of the tensor are zero when pulled back to the surface $\mathcal{H}$. 

Due to the product rule for the covariant derivative, we may study the covariant derivatives of the frame basis $\{ \ell_a,  n_a, m_a, \bar{m}_a \}$ and the frame derivatives of the non-zero components of the Ricci and Weyl curvature scalars separately to show that the b.w. +1 terms coming from these quantities vanish on the horizon $\mathcal{H}$. Denoting the intrinsic covariant derivative operator $\mathcal{D}$ on $\mathcal{H}$ induced by the spacetime derivative operator by $\nabla$, equations (B3-B6) in \cite{Ashtekar} imply: 

\beq &\mathcal{D}_a \ell^b \hat{=} \omega_a \ell^b = [(\alpha + \bar{\beta}) m_a + (\bar{\alpha} + \beta) \bar{m}_a - (\epsilon + \bar{\epsilon}) n_a] \ell^b & \nonumber \\
& \bar{m}^b \nabla_a n_b \hat{=} \lambda m_a + \mu \bar{m}_a - \pi n_a &  \\
& m^b \nabla_a \bar{m}_b \hat{=} - (\alpha - \bar{\beta}) m_a + (\bar{\alpha} - \beta) \bar{m}_a + (\epsilon - \bar{\epsilon}) n_a & \nonumber \eeq

\noindent where we have used $\hat{=}$ to denote 'equals on $\mathcal{H}$ to'. The sole terms that can contribute b.w. +1 terms is the spin coefficient $\epsilon$ and its complex conjugate, $\bar{\epsilon}$. 

For any non-extremal WIH with $\kappa_{(\ell)} = 2 \epsilon ~\hat{=} \text{ constant} $ and $\epsilon \hat{=} \bar{\epsilon}$, we apply a boost so that $\epsilon = 0$, and hence the horizon is extremal. This implies that the geodesic $\ell_a$ is now affinely parametrized. Relative to this frame, the covariant derivatives of the frame vectors on the horizon do not contribute any b.w. +1 terms. 

The positive b.w. components of the Weyl and Ricci tensors vanish when pulled back to the horizon, $\Phi_{00} \hat{=} \Phi_{10} \hat{=} \Psi_0 \hat{=} \Psi_1 \hat{=} 0$; therefore, these are globally zero on the horizon and the frame derivatives of these quantities are zero on $\mathcal{H}$.  The frame derivatives of the Ricci and Weyl curvature scalars which contribute b.w. +1 terms are $D \Phi_{02}, D \Phi_{11}, D \Psi_2$, and $D R$. The pull back of the Bianchi identities (BI-b), (BI-c), and (BI-i) in \cite{Stewart} yield

\beq D \Phi_{02} \hat{=} 0,~~ D \Psi_2 + 2 D \Lambda \hat{=} 0,~~ D \Phi_{11} + 3 D \Lambda \hat{=} 0,  \eeq

\noindent where $\Lambda = R/24$. The NP equations (NP-c), (NP-d), and (NP-e) in \cite{Stewart} give the following conditions on the spin coefficients: 

\beq D \tau \hat{=}0,~~~ D \alpha \hat{=} 0,~~ D \beta \hat{=} 0.  \eeq

\noindent Furthermore, on the null surface $\mathcal{H}$ the gauge choice $n_a = -dv$, where $\mathcal{L}_\ell v = 1$, implies

\beq \mu \hat{=} \bar{\mu},~~\text{ and } \pi \hat{=} \alpha + \bar{\beta}.  \eeq

As $\kappa = \rho = \sigma = 0$, and $\gamma \hat{=} 0$, we use the pullback of the NP equations (NP-f) and (NP-l) to solve for $\Phi_{11}$:

\beq 2 \Phi_{11} \hat{=} (\tau + \bar{\pi}) \alpha + (\bar{\tau} + \pi) \beta + \delta \alpha - \bar{\delta} \beta - \tau \pi - \alpha \bar{\alpha} - \beta \bar{\beta} - 2 \alpha \beta.  \eeq

\noindent Noting that $[D, \delta] \hat{=} 0$ and $\pi \hat{=} \alpha + \bar{\beta}$ implies that $D \Phi_{11} \hat{=} 0$. Therefore, $$D \Phi_{02} \hat{=} D \Phi_{11} \hat{=} D \Psi_2 \hat{=} D \Lambda \hat{=} 0$$ and it follows that the covariant derivative of the Ricci and Weyl tensor have at most non-zero b.w. 0 terms. 

To show that the second covariant derivative of the Ricci and Weyl tensors have vanishing b.w. +1 terms, we need only examine the frame derivative of the NP curvature scalars: 

\beq D^2 \Phi_{12},~~\text{ and } D^2 \Psi_3.  \eeq

\noindent Taking the frame derivative of the pullback of the Bianchi identities (BI-d) and (BI-j) we find that

\beq & D^2 \Psi_3 - D^2 \Phi_{21} \hat{=} 0 & \nonumber \\
& D^2 \Phi_{12} \hat{=} 0.  \eeq 

\noindent It follows that $D^2 \Psi_3 \hat{=} D^2 \Phi_{12} \hat{=} 0 $ and so $R_{ab;cd}$ and $C_{abcd;ef}$ both have at most non-zero b.w. 0 terms. 

Lastly, to show that the third covariant derivatives of the Ricci and Weyl tensor have vanishing b.w. +1 terms, we must show that 

\beq D^3 \Psi_4 \hat{=} D^3 \Phi_{22} \hat{=} 0.  \eeq

\noindent This can be  achieved by taking the pullback of (BI-g) and (BI-k). With the vanishing of these remaining positive b.w. terms on the horizon, any higher covariant derivative of the Weyl and Ricci tensors will have no positive b.w. terms, implying that they are of type {\bf II}.
\end{proof}

As a corollary of Theorem \ref{thm:WIH} we have the following result which guarantees a necessary condition for the detection of a WIH in terms of the vanishing of SPIs.

\begin{cor}
On any WIH, the discrimiant SPIs of the Riemann tensor and its covariant derivatives must vanish on this surface. 
\end{cor}

\noindent This is a necessary condition, but it is not sufficient due to the possibility that the entire set of discriminant SPIs could vanish on other surfaces besides the WIH.
%\subsection{Dynamical Black Holes}

\subsection{Dynamical Black Holes: Shear-Free Surfaces}

Co-dimension two spacelike submanifolds play a central role in gravitational theories based on Lorentzian geometry. Trapped, or marginally trapped submanifolds are examples of these submanifolds. The quasi-local horizons are submanifolds of co-dimension one (hypersurfaces) foliated by marginally trapped compact spacelike submanifolds of co-dimension two.  
%The property of being trapped is related to the volume change of the submanifolds, and is encoded in the sign of the divergence or ``expansion'' of given null normal vector fields.  Thus, the volume of closed (marginally) trapped submanifolds decreases  initially along every possible direction of future evolution. 
Unlike event horizons, trapping horizons and dynamical horizons depend on the foliation of spacetime, and are not invariant under conformal transformations \cite{FN2011}. This dependence on foliation implies they are highly non-unique \cite{AG2005}. The boundary $B$ of the future-trapped region $T$ of a spacetime (the region through which future-trapped surfaces pass) is foliation independent; however, it has been shown to be non-local, and is not a MTT \cite{Senov}.

In \cite{BS,BS1}, a preferred marginally trapped tube is proposed based on the determination of which region of the spacetime is absolutely indispensable for the existence of the black hole by sustaining the existence of closed trapped surfaces. A region ${Z}$ is called the {\em core} of the ``trapped'' region ${T}$ if it is a minimal closed connected set that must be removed from the spacetime in order to eliminate all closed trapped surfaces in ${T}$, and such that any point on the boundary $\partial{Z}$ is connected to $B=\partial {T}$ in the closure of the remainder \cite{Senov2}. While ${Z}\subset {T}$, in general ${Z}$ is substantially smaller than the corresponding trapped region ${T}$. The cores are not unique, although this non-uniqueness is often due to the existence of a high degree of symmetry of the spacetimes. In less symmetric cases the uniqueness of the cores cannot be assumed \cite{Senov}. 

%For example, in the case of spherically symmetric spacetimes, the surface $r=2m$ is the unique boundary of the spherically symmetric core. However, there exists non-spherically symmetric cores in the spherically symmetric spacetimes. 

%The `core' of a black hole is a minimal region that should be removed from the spacetime ${\mathcal M}$ in order to get rid of all possible closed trapped surfaces  \cite{Senov2}. 

%We need a physically sound criterion for selecting a preferred MTT.  Study shape tensor (depends on gauge).  null expansion -- shear scalars along ingoing/outgoing null directions vanish (shearfree/umbilical) [i.e.  assume there exists conditions on the shear at a prefered MTT.]  Proposal:  use shearless conditions as criterion to select preferred MTT [endo-shear-free ESF-MTT conjecture]; this is motivated from the  spherically symmetric analysis]].

In the case of spherically symmetric spacetimes the preferred MTT $r=2m$ happens to be the unique boundary of the spherically symmetric core; however, there are other non-spherically symmetric cores. While for the quasi-spherical Szekeres dust solutions \cite{sussman2012, Szekeres}, the preferred MTT is given by an apparent horizon $R = 2M$. A physically sound criterion is needed for selecting a preferred MTT. Another geometrical quantity associated with horizons and their evolution is the {\em shear} of the normal vectors of the surface \cite{Umbilical}. As an example, the standard horizons of isolated black holes in equilibrium have vanishing shear. The shear can be seen as a measure of the local instantaneous deformation of a given submanifold with its volume fixed as it starts to evolve. Mathematically, this is also associated to the extrinsic properties of the submanifold, and the shear-free property corresponds to the submanifold being umbilical along the evolution direction \cite{Umbilical, Senov2016}.

To discuss shear-free surfaces \cite{Umbilical}, we must introduce additional concepts. A spacelike submanifold is an orientable 2D Riemannian manifold $(S,g)$ immersed in a 4D Lorentzian manifold $(M,\bar{g})$. If $\bar{\nabla}$ and $\nabla$ are the Levi-Civita connections for $(M, \bar{g})$ and $(S, g),$ respectively, and $X,Y \in \mathcal{X}(S)$ are tangent to $S$ and $\zeta \in \mathcal{X}(S)^\perp$ is normal to $S$, then the relationship between the two are given by the formulas of Gauss and Weingarten: 
\beq & \bar{\nabla}_X Y = \nabla_X Y + K(X,Y), &  \\
& \bar{\nabla}_X \zeta = - A_\zeta X + \nabla_X^\perp \zeta, & \eeq
\noindent where $K(X,Y) = K(Y,X) \in \mathcal{X}(S)^\perp$ for all $X,Y \in \mathcal{X}(S)$ is the {\it second fundamental form} (or shape tensor) of the immersion, $A_{\zeta}$ is a self-adjoint operator called the {\it shape operator} associated to $\zeta$, and $\nabla^\perp$ is a connection in the normal bundle \cite{Oneill83}.

Demarcating the indices $A,B,\cdots$ as indices within $S$, the induced metric is now $\bar{q}_{AB}$. The second fundamental form may be treated as the following mixed index rank-three tensor $K^a_{~AB}$, while the shape operator becomes $A_\zeta = K^a_{~AB} \zeta_a$. We may define the {\it shear} scalars along the normal null vectors $k^+_a = \ell_a,$ and $k^-_a = n_a$ in terms of these quantities:

\begin{defn}
The shear along $k^\pm$ is given at the surface $S$ by 
\beq (\sigma^\pm)^2 = \left( K_{AB}(k^\pm) - \frac12 \theta^\pm \bar{q}_{AB} \right) \left( K^{AB}(k^\pm) - \frac12 \theta^\pm \bar{q}^{AB} \right) \eeq
\noindent where $K_{AB}(k^\pm) = A^a_{~AB} k_a^\pm$ and $\theta^\pm = \theta_{(k^\pm)}$. 
\end{defn}
 
\noindent A normal null direction, say $k^+$, is {\it shear-free} at $S$ if the corresponding shear scalar vanishes; $\sigma^+ = 0$. This can be generalized to non-null normal vectors as well.%As the surface $S$ is positive-definite, this gives an equivalent definition for a normal direction:

\begin{defn}
A spacelike surface is said to be {\it shear-free} along a normal direction $r^a$ if and only if the following condition 
\beq r_a K^a_{~AB} = \frac{1}{2} \theta_{(r)} \bar{q}_{AB} \nonumber \eeq
\noindent where $\bar{q}_{AB}$ is the first fundamental form and $K^a_{~AB}$ is the second fundamental form of the surface. 
\end{defn}
%
%\noindent is satisfied.
%
%If a surface $S$ has a normal null direction $\ell^{+}$ which is shear-free, this is equivalent to the second fundamental form $K_{AB}(k^+)$ being proportional to the first fundamental form, and the surface is  said to be umbilical along $\ell^{+}$. For a surface with normal null vectors $k^{+} = \ell $ and $k^{-} = n$, the shear-free condition implies \cite{Umbilical} 
%\beq [ K^+, K^- ] = 0, \nonumber \eeq
%\noindent where $K^\pm = K(k^\pm)$ and $[\cdot, \cdot]$ denotes the commutator.
%
%
%%Submanifold theory was explicitly exploited in \cite{Umbilical} and the extrinsic quantities associated with the deformation of submanifolds along normal directions was introduced, to study the relationship with their umbilical properties (in general semi-Riemannian manifolds for generic Riemannian submanifolds, keeping the dimension and co-dimension free). The \emph{total shear tensor} as the trace-free part of the second fundamental form tensor, the \emph{shear operators} of the trace-free parts of
%%the corresponding shape operators, and the {\em shear scalars}, were introduced. Several notions of umbilical submanifolds were then considered. In particular, the relevant case of co-dimension $2$ submanifolds in order to characterize spacelike co-dimension $2$ submanifolds that are umbilical along a normal direction has been investigated in detail.

Requiring that the MTT is foliated by 2D surfaces which are shear-free leads to a new quasi-local surface \cite{Senov2016}:

\begin{defn}
An {\it endo-shear-free Marginally Trapped Tube} (ESF-MTT) is a hypersurface foliated by shear-free marginally trapped surfaces.
\end{defn}

\noindent For spherically symmetric spacetimes, the round spheres are totally umbilical, and hence the spherically symmetric MTT is an ESF-MTT. For more general spacetimes, to determine if a particular shear-free MTS produces a unique ESF-MTT, one must deform the surface along normal directions and see if the shear-free property is preserved for one of the MTTs. In the case of an ESF-MTT, the vanishing of the variation of the expansion, $\delta_\zeta \theta^{+} = 0$, imposes additional conditions \cite{AMS2005,AMS2007} that depend on the  Gaussian curvature on $S$, the Einstein tensor, $G_{ab}$, the covariant derivative on $S$ and $$ W\equiv \left.G_{ab}k_{+}^a k_{-}^{b}\right|_S +\sigma^2, \label{W} $$ with $\sigma^2$ the shear scalar of $\vec k$ at $S$ \cite{Senov2016}.  Obviously $W\geq 0$ whenever $\left.G_{ab}k_{+}^a k_{-}^{b}\right|_S\geq 0$.  Assuming $\left.G_{ab}k_{+}^a k_{-}^{b}\right|_S \geq 0$, $W$ will vanish only if $\left.G_{ab}k_{+}^a k_{-}^{b}\right|_S=\sigma^2=0$, which leads to isolated horizons \cite{AshtekarKrishnan}. Given that the WIHs are detected by SPIs we conjecture that the preferred MTT in the dynamical regime will also be detected by SPIs.

\subsection{Imploding Spherically Symmetric Metric}

In the study  of spherically symmetric spacetimes, the variation $\delta_{\zeta} \theta_{sph}^{+}$ along normal directions simplifies drastically \cite{Senov, Senov2}, because $\sigma^2=0$ ($\vec k$ is a shear-free) and $s_B=0$, yielding a simplification to the variations of the second fundamental form and the expansion. It was shown that in spherically symmetric spacetimes there are closed ``trapped'' surfaces (topological spheres) penetrating both sides of the (non-isolated part of the) apparent horizon with arbitrarily small portions outside the region $\{r>2m\}$. Further, it has been shown that the MTT is unique \cite{BS,BS1}.

Let us consider the imploding spherically symmetric metric in advanced coordinates:

\beq ds^2 = - e^{2\beta(v,r)} \left(1 - \frac{2m(v,r)}{r}\right) dv^2 + 2 e^{\beta(v,r)}dv dr + r^2 d \Omega^2, \label{SphSymMtrc} \eeq

\noindent where $m(v,r)$ is the mass function and $\beta(v,r)$ is an arbitrary function. In this form all gauge freedom has been used and, in general, further simplification of the Einstein tensor is not possible; for example, for a perfect fluid solution the fluid (or dust) is not, in general, comoving.

The spherically symmetric metric is not an exact solution. However, it does contain exact solutions as special cases. For example, if $\beta =0$ and the matter source consists of null radiation we recover the imploding Vaidya exact solution, and the unique dust solution with $\beta_{,v} = 0$ is given by the LTB solution, which follows from the fact that the frame vectors on the horizon are geodesic and hence the horizon is ``geodesic-lined'' (see below) since  $\beta_{,v} =0$. 

We choose the two future pointing radial null geodesic vector fields:  

\beq \ell = \partial_v + \frac12 \left( 1 - \frac{2m}{r}\right) \partial_r ,~~n = e^{-\beta} \partial_r,  \label{SSln} \eeq
\noindent where $\ell_a n^a = -1$, and complete the non-coordinate basis using the complex spatial vector and its complex conjugate:
\beq m = \frac{1}{\sqrt{2}r} \partial_\theta + \frac{i}{\sqrt{2} r \sin(\theta) } \partial_\phi.  \label{SSmmb} \eeq
\noindent Relative to this basis the metric is now diagonalized: 

\beq g_{ab} = -\ell_{(a} n_{b)} + m_{(a} \bar{m}_{b)}. \label{gfrm} \eeq
 
The mean curvature vector for each round sphere (with $r$ and $v$ constant), is given by:
\beq H = \frac{2}{r} \left(e^{-\beta} \partial_v + \left(1 - \frac{2m}{r}\right) \partial_r\right). \nonumber \eeq
\noindent This implies that relative to the null frame, the future null expansions for these null vectors are:
\beq \theta_{(\ell)} = \frac{e^{\beta}}{r}\left(1-\frac{2m}{r}\right),~~~ \theta_{(n)} = - \frac{2e^{-\beta}}{r}. \eeq

The unique spherically symmetric FOTH is given by the surface $r-2m(v,r) = 0$, which is equivalent to $\theta_{(\ell)} = 0$ \cite{Senov}, and we will denote it as $\tilde{\mathcal{H}}$. This surface can be timelike, null or spacelike depending on the sign of the magnitude of the normal vector,
\beq n_a = \nabla_a(r-2m) = (1-2m_{,r}) dr - 2m_{,v} dv, \nonumber \eeq
\noindent and so evaluating the norm on the surface we obtain:
\beq |n| = g^{ab}n_a n_b = -4e^{-\beta}m_{,v}(1-2m_{,v})|_{\tilde{\mathcal{H}}}. \label{Normalnorm} \eeq

\noindent Assuming that $m_{,v} \neq 0$, the surface $\tilde{\mathcal{H}}$ will be spacelike or timelike and hence a dynamical horizon. 

Previously it was noted that the vanishing of  $m_{,v}$ on  $\tilde{\mathcal{H}}$ implies that the surface $\tilde{\mathcal{H}}$ is an isolated horizon \cite{Senov}. While this surface is indeed null, $\tilde{\mathcal{H}}$ will be an isolated horizon of interest only if the Einstein field equations are satisfied on $\tilde{\mathcal{H}}$. Thus, any spherically symmetric exact solution with $m_{,v} = 0$ will satisfy the geometric horizon conjecture. For example, the Weyl and Ricci tensors and their covariant derivatives are of algebraic type {\bf II} on the surface $r=2m$ for the ingoing Vaidya solution and the LTB solution (both of which admit an isolated horizon; i.e., $m_{,v} =0$). 

More generally there is a class of imploding spherically symmetric metrics which are not exact solutions, but emulate exact solutions like the Vaidya solution and the class of LTB solutions admitting isolated horizons. For these solutions, when $m_{,v} = 0$ on $\tilde{\mathcal{H}}$, the frame vectors normal to $\tilde{\mathcal{H}}$ are geodesic, leading to the following definition:

\begin{defn}
$\hat{\mathcal{H}}$ is a {\it geodesic lined horizon} \cite{AD} if 
\begin{itemize}
\item $\hat{\mathcal{H}}$ is diffeomorphic to the product $\check{\mathcal{H}} \times \mathbb{R}$ where $\check{\mathcal{H}} $ is a 2-sphere and the fibers of the projection: 
\beq \Pi : \check{\mathcal{H}}\times \mathbb{R} \to \check{\mathcal{H}} \nonumber \eeq
\noindent are geodesic curves. 
\item On each leaf $\hat{\mathcal{H}}$ the expansion $\theta_{(\ell)}$ of the null normal $\ell^a$ vanishes. 
\item The expansion $\theta_{(n)}$ of the null normal $n^a$ is negative. 
\end{itemize}
\end{defn}

\noindent For the class of spherically symmetric metrics admitting a null geodesic-lined horizon, and hence $m_{,v}=0$, the condition for $\ell^a$ and $n^a$ to be geodesic requires the following condition on $\beta$:

\beq \beta_{,v} = 0. \nonumber \eeq

With this condition, we may employ the NP formalism to show that the Weyl and Ricci tensors and their covariant derivatives are of algebraic type {\bf II} on the null surface $r=2m$ . 

\begin{prop} \label{thm:GLhorizons}
For any imploding spherically symmetric metric admitting a geodesic-lined null horizon, $\hat{\mathcal{H}}$, the Riemann tensor and its covariant derivatives  are of type {\bf II/D} on $\hat{\mathcal{H}}$.
\end{prop}

\begin{proof}
The proof of the proposition follows immediately by assuming $m_{,v} = 0$ on $\hat{\mathcal{H}}$. Then if $\Phi_{00} \neq 0$ the metric is of Weyl type {\bf D} and of general algebraic Ricci type {\bf I} everywhere except on the horizon where it is of type {\bf II} since the sole positive b.w. component $\Phi_{00}$ vanishes there. In the case that $m_{,v} = 0$ on $\hat{\mathcal{H}}$ and $\Phi_{00} =0$ we have, from condition \eqref{Phi000cond} below and applying a boost, that $\Phi_{22} = \frac{-(r-2m)^2 m_{,v}}{4r}$, wherein the Ricci tensor is of type {\bf II} everywhere except on the horizon where it is of type {\bf D}. To show that this also holds for all higher covariant derivatives, we apply a boost to fix $\epsilon = 0$; then the positive b.w.  terms of the covariant derivative of the Weyl and Ricci tensors are constructed from the frame derivatives: $D\Phi_{11}$, $D\Lambda$ and $D\Psi_{2}$, along with the non-zero spin-coefficients $\rho$ and $\mu$. It is clear that the spin coefficient $\rho$ vanishes on $r=2m$ from \eqref{SSrhomu}. As $\beta_{,v} = 0$ the frame derivatives  $D \Phi_{11}, D \Lambda,$ and $D \Psi_2 $ vanish when $r = 2m$. Therefore, the Ricci and Weyl tensors are of type {\bf II/D} on the horizon. It can also be shown that the higher order covariant derivatives of the Weyl and Ricci tensors are algebraically special by identifying the b.w. $+1$ terms that arise at each higher order and showing that they must vanish on the surface $r=2m$.

\end{proof}

\noindent If the metrics are solutions to the Einstein field equations, the null surface will be an isolated horizon. This result is consistent with the previous results of \cite{AD}.

However, in general, $m_{,v}$ need not be zero on $\tilde{\mathcal{H}}$. Let us consider the case with no further constraints from the field equations (i.e., this will not represent an exact solution) when the surface $r=2m$ is spacelike and hence a FOTH. We can apply the NP formalism to show that the covariant derivatives of the Weyl and Ricci tensors are no longer of type {\bf II} on the FOTH. However, relative to a particular frame, the structure of the covariant derivative of the Weyl tensor changes in a consistent manner on the surface $r=2m$. Additionally, the behaviour of the Weyl tensor's structure on the horizon can be defined in a frame-independent manner by the vanishing of a SPI.

Relative to the coframe given by \eqref{SSln} and \eqref{SSmmb}, the only non-zero component of the Weyl spinor is

\beq \Psi_2 &=& \frac{e^{-\beta} \beta_{,v,r}}{6} - \frac{m_{,r} \beta_{,r}}{2r} - \frac{r(r-5m)\beta_{,r}}{6r^3}+\frac{2 m_{,r}}{3 r^2} - \frac{m_{,r,r}}{6r} \label{SSpsi2} \\
& &+ \frac{r^2(r-2m) \beta_{,r}^2}{6r^3} -\frac{6m - r^2(r-2m) \beta_{,r,r}}{6r^3}, \nonumber \eeq

\noindent and so the parameters of the null rotations about $\ell_a$ and $n_a$ are fixed to identity. The remaining frame freedom consists of boost and spins. However, the boosts can be fixed at zeroth order as well since the non-zero NP curvature scalars for the Ricci spinor and the Ricci scalar $R = \Lambda/24$ are: 
\small
\beq & \Phi_{00} =  \frac{e^{2\beta} (r-2m)^2 \beta_{,r}}{4 r^3} + \frac{e^\beta m_{,v}}{r^2},&  \label{SSP00}\\
%& \Phi_{11} = \frac{e^{-\beta} \beta_{,v,r} }{4} +\frac{r(r-2m)\beta_{,r}^2}{4r^2} - \frac{3 \beta_{,r}}{4}\left(\frac{m}{r}\right)_{,r}+ \frac{2m_{,r} - m_{,r,r} r}{4r^2} + \frac{r(r-2m)\beta_{,r,r}}{4r^2},& 
& \Phi_{11} =  \left( \frac{\beta_{,r}^2}{4r} + \frac{\beta_{,r,r}}{4r} - \frac{3 \beta_{,r}}{8 r^2} \right) (r-2m) + \frac{\beta_{r,v} e^{-\beta}}{4} - \frac{3 \beta_{,r} m_{,r}}{4r} + \frac{3 \beta_{,r}}{8r} - \frac{m_{,r,r}}{4r} + \frac{m_{,r}}{2r^2}  & \label{SSP11}\\
& \Phi_{22} = \frac{e^{-2\beta} \beta_{,r}}{r}, &  \label{SSP22} \eeq

\beq  \Lambda = \left( - \frac{\beta_{,r}^2}{12 r} - \frac{\beta_{,r,r}}{12 r} - \frac{\beta_{,r}}{24 r^2}\right)(r-2m) - \frac{\beta_{,r,v} e^{-\beta}}{12} + \frac{\beta_{,r} m_{,r}}{4r} - \frac{\beta_{,r}}{8r} + \frac{m_{,r,r}}{12 r} + \frac{m_{,r}}{6r^2}  \label{SSLambda} \eeq
\normalsize

%\beq  \Lambda = \frac{R}{24} &=& -\frac{e^{-\beta} \beta_{,v,r} }{12} -\frac{r(r-2m)\beta_{,r}^2}{12r^2} + \frac{\beta_{,r}( 3r m_{,r} - \frac12 (r-2m))}{12 r^2}  \label{SSLambda} \\
%&& +\frac{2m_{,r}+r m_{,r,r}}{12 r^2} - \frac{r(r-2m)\beta_{,r,r}}{12 r^2}, \nonumber \eeq

\noindent The components of the Ricci spinor are related to the Ricci tensor by
\beq & \Phi_{00} = \frac{1}{2} R_{11},~~\Phi_{22} = \frac12 R_{22}, \Phi_{11} = \frac14 (R_{12} + R_{34}), \nonumber \eeq

\noindent with all other components of the Ricci tensor vanishing. Therefore, the Cartan scalars $\Phi_{00}, \Phi_{11}$ and $\Phi_{22}$ have b.w. $+2$,0,$-2$, respectively. Similarly, the nonzero Weyl spinor component is related to the only algebraically independent component of the Weyl tensor:

\beq \Psi_2 = - C_{1324}. \nonumber \eeq

Thus the Weyl tensor is of algebraic type {\bf D}, and the Ricci tensor is generally of algebraic type {\bf I} ($\Phi_{00} \neq 0$) relative to the alignment classification \cite{classa,classb,classc}. At zeroth order the isotropy group of the Riemann tensor consists of spins \footnote{In fact, the spins belong to the isotropy group of the metric, and so all higher covariant derivatives of the Riemann tensor are invariant under spins.}. The type {\bf II/D} SPIs for the spin invariant Ricci tensor are insensitive to the FOTH defined by $r= 2m(v,r)$. To detect the FOTH, we can investigate the covariant derivative of the Weyl tensor.

Taking the covariant derivative of the Weyl and Ricci tensors and applying the differential Bianchi identities, the components can be expressed in terms of $\Psi_2$, $\Phi_{00} = \Phi_{22}$,$\Phi_{11}$, $\Delta \Phi_{11}$ and the spin coefficients:
\beq & \epsilon = \frac{r(r-2m) e^\beta \beta_{,r}}{2r^2} + \frac{\beta_{,v} }{2} - \frac{e^\beta m_{,r}}{r} + \frac{m e^\beta}{2r^2}, & \label{SSeps} \\
& \rho = -\frac{e^\beta(r-2m)}{2r^2}, \text{ and }\mu =- \frac{e^{-\beta}}{r}. & \label{SSrhomu} \eeq

The covariant derivative of the Weyl tensor will have components of b.w. $+1, 0$ and $-1$. In particular, the non-zero components of b.w. $+1$ are: 

\beq C_{1214;3} = C_{1434;3} = C_{1213;4} = C_{1334;4} = 3 \rho \Psi_2, \label{SSBW10} \eeq
\noindent and 
\beq 2C_{1423;1} = C_{1212;1} = C_{3434;1} =  -4D \Lambda - 2\Delta \Phi_{00} -2 \mu \Phi_{00} + \rho (6 \Psi_2 + 4 \Phi_{11}), \label{SSBW01} \eeq
\noindent where we have simplified using the differential Bianchi identities.

Interestingly, the components in \eqref{SSBW10} vanish on the FOTH and hence identify the horizon. That is, this subset of algebraically special b.w. $+1$ components identify the horizon. However, the components in \eqref{SSBW01} do not; this can be seen explicitly by expressing the left-hand-side of \eqref{SSBW01} in coordinates. This implies that although the covariant derivative of the Weyl tensor is algebraically special on $\tilde{\mathcal{H}}$ it will not generally be of type {\bf II}. However, the existence of the FOTH clearly affects the structure of this tensor, and this condition on the structure of the Weyl tensor is reflected in the vanishing of a SPI:

\begin{thm}
For any spherically symmetric metric, the structure of the covariant derivative of the Weyl tensor changes on the FOTH $r=2m(v,r)$ and this can be detected by the invariant: 

\beq 4 I_1 I_3  - I_5,  \eeq 

\noindent where $I_1 = C_{abcd} C^{abcd},~~I_3 = C_{abcd;e} C^{abcd;e}$ and $I_5 = I_{1,a}I_1^{~a}$. 
\end{thm}

\begin{proof}
To prove this we may compute the explicit forms of the invariants $I_1, I_3$ and $I_5$ and combine them to show that

\beq 4 I_1 I_3  - I_5 =  2^{12} 3^3 \rho \mu \Psi_2^4 = \frac{2^{11} 3^3 (r-2m) \Psi_2^4}{r^3}. \eeq

\noindent Since $r-2m = 0$ on $\tilde{\mathcal{H}}$, the invariant vanishes on $\tilde{\mathcal{H}}$.
\end{proof}

\noindent Noting that $I_1 = 48 \Psi_2^2$ \cite{GANG}, we may normalize the above SPI to produce an invariant whose vanishing is necessary and sufficient to detect the FOTH. The resulting invariant will be proportional to $\rho$ which is a Cartan invariant  relative to the coframe which vanishes on the FOTH. The behaviour of the b.w. $+1$ components of the first covariant derivative of the Weyl tensor and the positive b.w. components of the second covariant derivative of the Weyl tensor can be investigated symbolically where it can be shown that, in general, only a subset of the highest b.w. terms of the covariant derivatives vanish on the surface $r=2m$. Each of the components in the subset of highest b.w. terms that vanish on the horizon may be expressed in terms of $\rho$. 

In general $m_{,v}$ is not necessarily zero on $\tilde{\mathcal{H}}$. However, additional constraints (field equations) will be imposed for physically realistic exact solutions (such as the Vaidya and LTB exact solutions discussed above).  To study the subset of imploding spherically symmetric metrics which satisfy the Einstein field equations, we can impose conditions on the components of the Einstein tensor relative to the coframe basis, $\{ n_a, \ell_a, \bar{m}_a, m_a \}$, which diagonalizes the metric \eqref{gfrm}, which are:

\beq G_{11} &=& \frac{e^{2\beta} \beta_{,r} (r-2m)^2}{2r^3} + \frac{2 e^\beta m_{,v}}{r^2} \nonumber \\ 
G_{12} &=& -\frac{\beta_{,r} (r-2m)}{r^2} + \frac{2m_{,r}}{r^2} \nonumber \\
G_{22} &=& \frac{2e^{-2\beta} \beta_{,r}}{r} \nonumber \\
G_{34} &=& -\frac{(-2\beta_{,r}^2 + \beta_{,r} - 2 \beta_{,r,r} r )(r-2m)}{2r^2} - \frac{-2e^{-\beta} \beta_{,r,v} r^2 + 6 \beta_{,r} m_{,r} r - \beta_{,r} r+ 2m_{,r,r} r }{2r^2}. \nonumber  \eeq

\noindent Equivalently (and for reference) the components of the Einstein tensor relative to the coordinate basis are:
\small
\beq G_{vv} = \frac{2 e^{2\beta} m_{,r} (r-2m)}{r^3} + \frac{2e^\beta m_{,v}}{r^2},~~ G_{vr} = e^{2\beta} G_{22},~~ G_{rr} = \frac{2\beta_{,r}}{r},~~ G_{\theta\theta} &=& G_{\phi \phi} = r^2 G_{34}. \nonumber \eeq%\left( \beta_{,r}^2 r -\frac12 \beta_{,r} + \beta_{,r,r} r \right)(r-2m) + e^{-\beta} \beta_{,r,v} r^2 - 3 r \beta_{,r} m_{,r} + \frac{3}{2}\beta_{,r} r - r m_{,r,r} .\nonumber \eeq 
\normalsize
\noindent We can consider several possible matter fields by imposing conditions on the NP Ricci scalars \cite{kramer} : 

\begin{itemize}
\item Null Radiation (Null Electromagnetic (EM) Field):

\beq \Phi_{00}=\Phi_{11} =0,~~ \Phi_{22} \neq 0  \eeq

\item Non-Null EM field:

\beq \Phi_{00}=\Phi_{22} = 0,~~ \Phi_{11} \neq 0  \eeq

\item Perfect Fluid:

\beq \Phi_{00} \Phi_{22} = 4 (\Phi_{11})^2 \label{perffluid} \eeq

\item Dust Solution 

\beq  \Phi_{00} \Phi_{22} = 4 (\Phi_{11})^2, 2\Phi_{11} =  \frac{R}{4} \nonumber \eeq

\end{itemize}
\noindent where $R = 24 \Lambda$ (this corresponds to the LTB solution).
\vspace{5 mm}

In the case of null or non-null radiation, $\Phi_{00} = 0$, and so \eqref{SSP00} implies that the surface $r = 2m$ is null; i.e., $m_{,v} = 0$. In the case that $\Phi_{00} =0$, the Ricci tensor is of type {\bf II}  and $\beta$ satisfies 
\beq (e^\beta )_{,r} = \frac{-4 m_{,v} r}{(r-2m)^2}. \nonumber \eeq\noindent If $\Phi_{11}=0$, then the EM field is null, and a boost can always be made so that the negative b.w. term of the Ricci tensor is of the form:
\beq \Phi_{22} = \frac{-(r-2m)^2 m_{,v}}{4r}, \label{Phi000cond} \eeq
\noindent 
\noindent at the cost of $\ell$ no longer being affinely parametrized. Thus, in this case, the Ricci tensor is of type {\bf III} everywhere except on the horizon where it is of type {\bf 0}. If this is an exact solution of Einstein's field equations then this will be an isolated horizon, and the covariant derivatives of the Weyl and Ricci tensors will be algebraically special on the isolated horizon, and the conjecture is satisfied. 

To further study the behaviour of an imploding spherically symmetric metric satisfying the Einstein field equations, we consider a perfect fluid solution and explicitly exclude the case of dust  (i.e., the exact LTB solution). The condition \eqref{PerfFluid} for a perfect fluid implies that the scalar, $I \equiv \Phi_{00}\Phi_{22} - 4\Phi_{11}^2$, must vanish. In coordinates this implies

\beq 
 I &=&  \left(\frac{e^{2\beta} (r-2m)^2 \beta_{,r}}{4 r^3} + \frac{e^\beta m_{,v}}{r^2} \right) \left(\frac{e^{-2\beta} \beta_{,r}}{r} \right) \label{PerfFluid}  \\ && - \left[ \frac{e^{-\beta} \beta_{,v,r} }{4} +\frac{r(r-2m)\beta_{,r}^2}{4r^2} - \frac{3 \beta_{,r}}{4}\left(\frac{m}{r}\right)_{,r}+ \frac{2m_{,r} - m_{,r,r} r}{4r^2} + \frac{r(r-2m)\beta_{,r,r}}{4r^2}\right]^2,   \nonumber \eeq

\noindent must vanish. We may derive additional conditions by imposing $I_{,r} = I_{,v} = 0$. 

However, we will focus here on a subclass of perfect fluid solutions by assuming that $\beta_{,v} = 0$ on the horizon. Evaluating $I = 0$ on the surface $r=2m$ gives an expression for $m_{,v}$ on this surface: 

\beq 
& I_0 \equiv  \frac{e^{-\beta} \beta_{,r} m_{,v}}{r^3} - \left[ - \frac{3 \beta_{,r}}{4}\left(\frac{m}{r}\right)_{,r}+ \frac{2m_{,r} - m_{,r,r} r}{4r^2} \right]^2 = 0. & \label{I0}  \eeq

\noindent Calculating $C_{1212l;1}$ evaluated on $r=2m$ (with $\beta_{,v} = 0$) we get:

\beq 12 r^5 C_{1212;1}\displaystyle|_{r=2m} &=& 4 \beta_{,r}^2 m_{,v} r^4 + 6 \beta_{,r} m_{,rv} r^4 + 4 \beta_{,rr} m_{,v} r^4 \nonumber \\ && - 10 \beta_{,r} m_{,v} r^3 + 2m_{,rrv} r^4 - 8 m_{,rv} r^3 + 12 m_{,r} r^2.  \label{dagger} \eeq

\noindent By evaluating $I_{,r}=0$ and $I_{,v}=0$ on the surface $r=2m$, we can derive expressions for $m_{,vr}$ and $m_{,rrv}$ which will allow us to represent this expression as: 

\beq C_{1212;1}|_{r=2m} = \{~\}m_{,v} + \{~\}(m_{,v})^\frac12 + \{~\}\frac{m_{,vv}}{(m_{,v})^\frac12 }. \nonumber \eeq
\noindent If $m_{,v} =0$ and $\frac{m_{,vv}}{\sqrt{m_{,v}}} =0$ on $r=2m$, it follows that $C_{1212;1} = 0$ there.

Let us compute $m_{,v} = f(v,r)$ close to a point on the FOTH, $$r=2m \equiv 2m_0.$$ In a neighbourhood $U$ of $r=2m_0$, we may set $v =0$ at the point and allow $v$ to vary in $U$. We will assume that as $r$ varies off the horizon, $r-2m_0$ is small, and that $m_{,v}$ is bounded in $U$, allowing for an analytic expansion. Thus, we may write a local expression for $m$ on U as:

\beq (m-m_0) = \frac12 r' + (m_1(r')v + m_2(r') v^2 + \cdots ) \nonumber \eeq

\noindent where $r' \equiv r-2m_0$ is small\footnote{Essentially in these coordinates, $r' = 0$ on the horizon, as in near horizon calculations \cite{Kunduri2013}. }. From this we find that $m_{,v}|_{r' =0} = m_1(r' = 0)$, so that $m_{,v} \cong m_1(r)$ on $U$. It can be then shown that  $m_1(r) \approx 0$ on $U$. Assuming $\beta_{,r} \neq 0$, and using \eqref{I0}, we may solve  for $\beta_{,r}$ perturbatively  and integrate to determine a form for $\beta$. Substituting $\beta$ into \eqref{dagger} then yields: 

\beq C_{1212;1} \sim (r-2m)^2 + {\cal O} [(r-2m)^3]. \eeq

From this we can conclude that for a well-behaved perfect fluid spherically symmetric solution,  in a neighbourhood of a point on the horizon, locally all of the b.w. $+1$ terms of $C_{abcd;e}$ are zero on the horizon. We emphasize that this is not a proof but rather an argument indicating how the development of a dynamical horizon into an isolated horizon occurs in a smooth manner, and correspondingly how the geometric horizon conjecture may be applied to dynamical horizons.

\subsection{Further Work} \label{subsec:FW}

% Make consistent 2D, 3D, 4D throughot paper.

In future work we will study the geometric horizon conjecture for less idealized dynamical black holes. First, we would complete a more comprehensive analysis of spherically symmetric models and further explore the transition from dynamical to isolated horizons. In particular, the behaviour of the geometric horizon could be studied in piecewise linear Vaidya and LTB solutions to determine whether the evolution of this surface is smooth at all times. Additionally, the collapse of matter into a black hole can be generalized to configurations which are not spherically symmetric. For example, the quasi-spherical Szekeres dust models are a generalization of the LTB solutions representing collapsing non-concentric shells of matter which admit apparent horizons \cite{Szekeres}. To extend this to arbitrary dynamical black holes we could consider a generic metric arising by perturbatively reconstructing spacetime near the horizon \cite{booth} and examine the behaviour of the discriminant SPIs for the Ricci and Weyl tensors. These invariants will be large, but it is possible that they can be simplified by imposing additional conditions on the horizon such as extremality or axisymmetry.

Coalescing black holes provide another scenario where dynamical horizons appear, and hence provide a test for the geometric horizon conjectures. While a numerical simulation of physically realistic black holes coalescing is unavailable to examine the algebraic type of the curvature tensor on the horizon, there are exact solutions that can be studied, such as the Kastor-Traschen solution \cite{KT}. The Kastor-Traschen solution represents $N$ charge-equal-to-mass black holes in a spacetime with a positive cosmological constant, $\Lambda$:  
%The corresponding spacetime metric has the following form (see \cite{NSH}):
\begin{eqnarray}
& ds^2=-W^{-2}dt^2+W^2(dx^2+dy^2+dz^2)\,;~W=-H\tau+\displaystyle\sum_{i=1}^N\frac{m_i}{r_i}. & 
\end{eqnarray}
Here $H=\sqrt{\Lambda/3}$, where $\Lambda \ge 0$ is the cosmological constant, $t\in(-\infty,0)$, $m_i$ $i\in [1,N]$ are the black hole masses, and $r_i=\sqrt{(x-x_i)^2+(y-y_i)^2+(z-z_i)^2}$, are the black hole positions where $r_i = 0$, $i \in[1, N]$, represent a 3D infinite cylinder, with 2D cross-sectional area of $4\pi m_i^2$ for each black hole. The electromagnetic 4-potential is $A=W^{-1}dt$.

The existence of horizons has been examined in the case of two coalescing black holes \cite{NSH}. In this case, we can choose coordinates so that the black holes are located on the $z$-axis at the coordinate distance $c>0$ from the origin,
\begin{eqnarray}
r_{\pm}&=&\sqrt{x^2+y^2+(z\pm c)^2}.
\end{eqnarray}

%xxyyzz Note for Alan about W_1 and W_2 pg 30 of edits

\noindent If the sum of the masses of the two black holes is below a critical mass, the black holes will coalesce into a larger single black hole. For these spacetimes the SPI $\mathcal{W}_2$ vanishes while $\mathcal{W}_1$ is generally non-zero\footnote{The type {\bf II}/{\bf D} SPIs $\mathcal{W}_1$ and $\mathcal{W}_2$ are constructed from traces of the powers of the Weyl tensor and are defined by equations \eqref{weyl1} and \eqref{weyl2}.}.  At earliest times, ${\cal W}_1 \to  0$ as $t \to-\infty$, and there are two 3D geometric horizons enclosing the two black holes.  

In the  exact equal mass Kastor-Traschen solutions the type  {\bf II}/{\bf D} discriminant  ${\cal W}_1$ vanishes on segments of the symmetry-axis, at  the coordinate locations of the black holes $r_{\pm}=0$  and on a ``dynamical" 2D (cylindrical) surface  around the symmetry-axis that appears in the center of mass plane \cite{AD}. At earlier stages of the coalescence, this 2D surface has a finite cross-sectional radius (from the symmetry-axis), while at later stages this surface expands as the two black holes move together. Using a measure of the separation between the black holes introduced in \cite{NSH},  it may be shown that as $t \to 0^- $ this measure approaches zero as the two black holes merge and the 2D surface forms around the two black holes,  suggesting that we can identify the location of a geometric horizon in the dynamical regime. After the black holes have merged, the spacetime will eventually settle down to a type {\bf{D}} Reissner-Nordstrom-de Sitter black hole of mass $m_1+m_2$, since ${\cal W}_1 \to  0$ as  $t \to0^-$, and in the quasi-stationary regime there will be a single 3D horizon \cite{GANG, AD}.

The type {\bf II/D} discriminant SPIs for the trace-less Ricci tensor and any trace-less operator produced from the covariant derivatives of the Ricci tensor will also vanish on 3D surfaces at a finite distance from the axis of symmetry.  This implies that the trace-less Ricci tensor and its covariant derivatives are of type {\bf II/D} on these surfaces. In addition, there is numerical evidence for a minimal 3D geometric surface where the invariant $\mathcal{W}_1$ takes on a constant non-zero minimum value and evolves in time. These results suggest the existence of a geometric horizon in the case of the dynamical regime of the Kastor-Traschen spacetime. A comprehensive analysis will be provided in future work, which will also consider two unequal black holes, and various configurations of three black holes in this class of exact spacetimes \cite{AD}.

\newpage

\section{Conjectures}

The  SPIs constructed by Page and Shoom \cite{PageShoom2015} that detect the stationary horizon are well-motivated for stationary black holes, as any compact Cauchy horizons admitted by real analytic vacuum solutions to the 4D Einstein equations must necessarily be Killing horizons when the horizon generating null geodesics are all assumed to be closed curves. Additionally, if a horizon generating Killing vector field exists in cases for which the null generators are not all closed then one or more additional Killing vector fields must also exist, which generate a certain (commutative) action of isometries of the spacetime \cite{Isenberg}. In dimension $D \geq 4$, it was proven that if a stationary, real analytic, asymptotically flat vacuum black hole spacetime contains a non-degenerate horizon with compact cross sections that are transverse to the stationarity generating Killing vector field then, for each connected component of the black hole's horizon, there is a Killing vector field which is tangent to the generators of the horizon \cite{Isenberg}. For the case of rotating black holes, the stationarity generating Killing vector field is not tangent to the horizon generators and therefore the isometry group of the spacetime is at least two dimensional. 
%This shows the laws of black hole thermodynamics can be extended to apply, in some sense, to black objects in higher dimensions; e.g., whether a higher dimensional black object, in the stationary case, necessarily admits a well defined, constant, surface gravity and corresponding Hawking temperature associated to its event horizon. 

%Asymptotic flatness plays a rather technical role in the analysis to ensure the theorem; but the methods could certainly be applied to more general problems in which there is a breakdown of asymptotic flatness. The qualification {\em{non-degenerate}}, means that the results do not necessarily apply to extremal black holes

Although this concept can be extended to dynamical black holes conformally related to stationary black holes \cite{PageMcNutt,McNutt2017}, it is too restrictive for most dynamically evolving black holes, as the entire spacetime history must be known to characterize an event horizon.
%, it is not well-suited for a dynamical description of black hole evolution
A less restrictive notion of a black hole horizon is given by a quasi-local isolated horizon, which is a specialization of a NEH, which accounts for equilibrium states of black holes and encompasses all essential local features of an event horizon \cite{AshtekarKrishnan,DiazPolo}.  Isolated horizons employ only local time-translational Killing vector fields and do not require asymptotic structures nor foliations of spacetime. We note that every Killing horizon that has the topology $S_2 \times R$ is an isolated horizon \cite{AshtekarKrishnan}. This implies that, in particular, the event horizon of Kerr geometry is an isolated horizon. The rotating Kerr-type isolated horizons in the astrophysical context can account for, to a good approximation, supermassive spinning black holes and the associated accretion disks in active galactic nuclei.  It is expected that the majority of physical black holes actually appear to rotate rapidly around an internal symmetry axis.

Working with a NEH avoids many of the problems that beset horizon identification (including, for example, locally determinability, uniqueness, and smoothness) as we only use quantities that are intrinsic to the horizon. In 4D, and assuming the usual energy inequalities, the existence of an induced degenerate metric tensor locally identified with a metric tensor defined on the 2D tangent space and its induced covariant derivative, leads to the condition that on the NEH the Weyl tensor is at most of algebraic type {\bf II}. 
%(where the aligned null vectors tangent to the surface correspond to a double principal null direction of the Weyl tensor)
As any NEH can be given the structure of a WIH, we conclude  {\em{that for a WIH, this implies that the Weyl tensor is of type {\bf II}/{\bf D} on the horizon}}  \cite{Lewandowski}.

Furthermore, all of the known exact higher dimensional black holes are algebraically special of Weyl (curvature) type {\bf II} or {\bf D} \cite{Pelavas}. This has led to a conjecture that asserts that
stationary higher dimensional black holes (with the additional conditions of vacuum and/or asymptotic flatness) are necessarily of Weyl type {\bf D} \cite{Pelavas}.  This conjecture has been supported by the study of local non-expanding null surfaces in $n$ dimensions \cite{Lewandowski}.  Assuming the usual energy inequalities, it was found that the vanishing of the expansion of a null surface 
%implies the vanishing of the shear allowing a covariant derivative operator induced on each non-expanding null surface.  The induced degenerate metric tensor, locally identified with a metric tensor defined on the $n-2$ dimensional tangent space, and the induced covariant derivative, locally characterized by the rotation 2-form in the vacuum case, constitute the geometry of a nonexpanding null surface.  The remaining components of the surface covariant derivative lead to constraints on the induced metric and the rotation 2-form in the vacuum extremal isolated null surface case.  This 
leads to the condition that on the NEH the boost order of the null direction tangent to the surface is at most 0, so that the Weyl tensor is at most of algebraic type {\bf II} (where the aligned null vectors tangent to the surface correspond to a double principal null direction of the Weyl tensor in the 4D case). Therefore, the corresponding WIH arising from the NEH must be of type {\bf II}

In the case of a star collapsing to form a black hole, it is expected that the exterior of the black hole will settle down to a stationary state eventually, such as the axisymmetric Kerr solution. If the exterior settles down to the Kerr metric, it is expected that due to continuity there will be a region of the interior near the horizon that should be close to the interior Kerr metric despite what the interior of the black hole settles down to. Within the black hole event horizon the Kerr metric has an inner horizon which is also a null surface. However, this inner horizon is unstable, and so for a spacetime that begins close to the Kerr metric the inner horizon may be something else, possibly even a singularity \cite{MassInflation}. Such a singularity is believed to maintain the inner horizon's character as a null surface, and this has been reinforced by a variety of analytic arguments, mathematical results, and numerical simulations\footnote{Note that a black hole with charge but no spin is described by the Reissner-Nordstrom metric, which is spherically symmetric like Schwarzschild but has an unstable inner horizon like Kerr.} \cite{Ashtekar,AshtekarKrishnan}.

These arguments support the notion that at later times the horizon is smooth and unique, and in principle might be identified by algebraic/geometrical conditions. Motivated by the  results for stationary horizons and NEHs, we assume that there is geometrically defined unique, locally determinable, smooth (dynamical) horizon for which the  curvature tensor is algebraically special due to the vanishing of the  expansion of a preferred null congruence which shields all other horizons, and identifies the region of interest. These horizons can be identified and located by SPIs, which are gauge invariant (and, in particular, are not dependent on spacetime foliations). This will not necessarily work in all possible cases but we do expect it to work in generic physical collapse, black hole coalescences, and exact black hole solutions. It is possible that the invariants may also vanish at fixed points of any isometries and along any axes of symmetry, and hence do not specify the horizon completely. However, we expect that identifying a smooth surface for physical situations is always possible, and that as we follow this unique, smooth surface back in time (during the physics of collapse or merger) this surface may suffer a bifurcation and may no longer unique or smooth (or even differentiable). 

%It is plausible that there exists a unique, smooth geometric horizon that shields all other horizons (or at least identifies the region of interest). 
%
%\begin{defn}
%A geometric horizon is defined as an algebraically special hypersurface where magnitudes of the discriminant SPIs take their smallest, possibly zero, values. These invariants detect where the Ricci and Weyl tensor are close to or of type {\bf II/D}.
%\end{defn}
%
%\noindent With this definition, we propose the following suite of conjectures 

\subsection{The Geometric Horizon Detection Conjectures}

For the black holes we have considered here, the horizon is always more algebraically special than other regions of spacetime, and if the Riemann tensor is of algebraic type {\bf II/D}, then so are the Ricci and the Weyl tensor in the same frame. To state the conjectures, we will say a tensor ${\bf T}$ is {\em $n^{th}$-order algebraically special} if ${\bf T}$ and all covariant derivatives of ${\bf T}$ up to order $n$ are of algebraic type {\bf II} or more special. 
  
%(i.e., they are aligned).
 
\paragraph{Conjecture Part I:}
{\em{If the whole spacetime is zeroth-order algebraically general, then on the horizon the spacetime is algebraically special, and this can be identified using SPIs.}}
\vspace{4 mm}

This part of the conjecture is more practical, and will conceivably be of use to numerical relativists who study the collapse or merger of real black holes, which are typically of general algebraic type away from the horizon. We expect that the conjecture might be qualitatively different for dynamical horizons, and that the vanishing condition might perhaps be replaced with a ``minimal condition''.  Since such a condition is defined in terms of an invariant quantity, it may (or may not) be possible to define it in a foliation independent manner.

\paragraph{Conjecture Part II:}
{\em{If the whole spacetime is zeroth-order algebraically special (and on the horizon the spacetime is thus also algebraically special) and if the  whole spacetime has an algebraically general first order covariant derivative of the Riemann tensor, $R_{abcd;e}$, then on the horizon $R_{abcd;e}$ will be  algebraically special and we can identify this surface using SPIs.}} 
\vspace{3 mm}

This can be repeated for higher covariant derivatives of the Riemann tensor if necessary. This part of the conjecture is more theoretical, necessitating analytic calculations, and can be applied to exact solutions. It may not be desirable to require in a general spacetime that the covariant derivatives be algebraically special (i.e., of type  {\bf II} or {\bf D}) to each order (i.e., of type $D^k$) on the black hole horizon as this might be too restrictive, since these spacetimes are necessarily degenerate Kundt or locally homogeneous \cite{Kundt,typedk}. 

%xxyyzz Ask about Lode being included. 
\newpage

\section{Discussion}

In this paper we have proven that the horizon for a stationary black hole can be detected by the vanishing of discriminant SPIs. Furthermore, by considering NEHs (which generalize Killing horizons) we have shown that the Ricci and Weyl tensors and their covariant derivatives become more special in terms of their algebraic type on this surface and consequently may be identified in terms of the vanishing of discriminant SPIs. Motivated by this result we have introduced the concept of a geometric horizon, which is a surface defined by conditions on (i.e., vanishing of) certain SPIs. In the case of a NEH, the discriminant SPIs vanish on this particular surface, implying that it is a geometric horizon. We then considered the class of imploding spherically symmetric metrics, which contain simple dynamical black hole solutions describing spherically symmetric matter falling towards a black hole, such as the Vaidya and LTB solutions, and showed that the unique spherically symmetric dynamical horizon can be defined by the vanishing of a SPI, and hence is foliation independent. We are currently investigating how the geometric horizon conjecture can be applied to dynamical horizons in more general spacetimes.

However, in physical problems with dynamical evolution, the horizon might not be unique or may not exist at all, and amendments to the conjecture may be necessary (e.g., replace vanishing conditions with minimum conditions). In order to make the definition of a geometric horizon more precise we need to focus on physical black hole solutions, and in order to prove definitive results we must append some physical conditions to the definition such as, for example, energy conditions, a particular theory of gravity, and some asymptotic conditions. The  geometric horizon conjectures are expected to apply in higher dimensions \cite{GANG}; however, we are primarily interested in applications in 4D. In higher dimensions, by algebraically special we mean algebraic type {\bf II} or more special, and by algebraically more general we mean of algebraic type more general than type {\bf II}.

In 4D, the uniqueness property of the Kerr solution has led to the belief that the stationary equilibrium limit reached in the evolution of an isolated system undergoing gravitational collapse (i.e., once the gravitational wave content has been radiated away) will eventually settle down to the exterior region of a Kerr geometry. This is believed, despite the fact that the canonical definition of a future event horizon refers to future null infinity and requires the existence of a global time translational Killing vector field for static spacetimes and an asymptotic time-translational Killing vector field at spacelike infinity for stationary spacetimes.  A number of {\em local and invariant} criteria have been proposed which act as {\em quality factors} to measure the deviation of a given stationary spacetime from the Kerr spacetime \cite{GPS}. Such criteria are formulated in terms of scalar quantities which  are not SPIs, and are only determined for spacetimes for which there exists a timelike Killing vector.  It is plausible that the geometric horizon could be employed by specifying initial data in numerical simulations of a radiating isolated system to test whether in the asymptotic regime it is close to the Kerr solution. In particular, by treating such a system at late times as a perturbed Kerr metric, if there is a well-posed initial-value problem we could feasibly integrate back and find a unique hypersurface. 
%A local characterisation of the Kerr solution was given in \cite{MARS-KERR}, using the Mars-Simon  tensor $\mathcal{S}_{abcd}$, to determine when a smooth vacuum or a non-flat Ricci-flat spacetime is locally isometric to the Kerr spacetime. The {spacetime Simon tensor} and the scalar constructed exclusively from the {superenergy density} can also be used to obtain an alternative local characterization of the Kerr spacetime \cite{GPS}.

In numerical relativity the generation of gravitational waveform templates for gravitational wave data analysis has been helpful in understanding the behaviour of black holes. The calculation of gravitational wave signals in the theoretical modelling of $D = 4$ sources in the framework of GR is well understood \cite{T, Bishop}, and is one of the most important diagnostic tools for studying the strong-field dynamics of compact objects in 4D spacetimes, as illustrated by the LIGO detection of GW150914 \cite{LIGO}.  The wave extraction techniques presently used in numerical simulations of astrophysical gravitational wave sources can be classified using the Weyl scalars from the NP formalism in 4D, while in higher dimensions a numerical implementation using the  formalism developed by  Godazgar and Reall \cite{Godazgar} and a generalisation of the NP formalism for which the Weyl tensor in higher dimensions is decomposed allows for gravitational wave extraction from numerical simulations of rapidly spinning objects in higher dimensions \cite{cook}. To employ the NP approach, the frame must be completely fixed  so that the resulting ``Weyl scalars'' are Cartan scalars. This has not yet been implemented for higher order (derivative) SPIs using the NP formalism.  

Of course, if our goal is to provide results that could be useful to numerical relativists, computability is an important issue.  And in this regard Cartan invariants have an  advantage over the related SPIs \cite{GANG,AM2016}. If a numerical calculation has already been done, and a postiori we want to investigate where a particular horizon or surface is located in the numerical solution, then for most computations involving Cartan invariants would be difficult (since Cartan invariants depend in principle on a particular choice of frame). However, we could use the same frame used in numerical simulations once the Cartan algorithm has been implemented and relate the resulting frame to the frame used in numerical simulations. It is possible a theory of approximate equivalence can be developed, giving a topology on the space of metrics. 

Whether or not our conjectures are useful will have to be evaluated. We have attempted to support the conjectures with some analysis and the study of some practical examples, but further work is required and perhaps additional refinement of the conjectures will be necessary. In a sense the conjectures refer to ``peeling'' properties (of the geometrical curvature) close to the horizon; i.e., the curvature is of algebraically special type {\bf II} close to the horizon. It is possible that as gravitational wave modes (of algebraic types {\bf III} and {\bf N}) dissipate to infinity, the {\em{horizon eventually settles down to be type {\bf D}}} under some reasonable asymptotic conditions.

In future work we will investigate the geometric horizon conjectures for exact solutions describing dynamical black holes such as black hole mergers and in-falling matter into a black hole. To examine the dynamical horizon arising from in-falling matter into a black hole, we will study the quasi-spherical Szekeres dust models, which represent collapsing shells of matter, but are not spherically symmetric like the LTB solutions \cite{Szekeres}. Using a generic metric admitting a horizon structure \cite{booth}, we can study the behaviour of the discriminant SPIs near an arbitrary dynamical horizon. We will also provide a comprehensive analysis of the geometric surfaces discussed in subsection \ref{subsec:FW} for  the Kastor-Traschen solution with two or three charge-equal-to-mass Reissner-Nordstr\"{o}m black holes \cite{KT}. The conjectures can be further tested for the exact solution arising from the binary black hole merger in the extreme-mass ratio \cite{Emparan2016}. Finally, we would like to consider the possibility of imposing initial data on a geometric horizon, and whether a well-posed initial value problem can be implemented for numerical simulations of dynamical black holes.

\newpage 
\section{ Appendix A: 4D Example} \label{4Dexa}

While the discriminant analysis provides necessary conditions for the curvature tensor or its covariant derivatives to be of algebraic type {\bf II/D}, they are difficult to compute in practice. The frame approach provides a direct confirmation of the type {\bf II/D} property of the horizon by explicitly constructing the frame in which the Ricci or Weyl tensors become algebraically special on the horizon. We employ the Cartan algorithm to determine the relevant frame \cite{GANG, CKAHD}. For a black hole solution, a frame may be chosen  where the curvature tensor or its covariant derivatives are of algebraic type {\bf I} outside of the horizon. Relative to this frame, these tensors  will become type {\bf II} or more algebraically special on the horizon. 

To illustrate this, we review the Kerr-Newman-NUT-(Anti)-de Sitter solution which admits stationary horizons \cite{GANG}. Since a stationary horizon is a special case of a WIH, we will show that the covariant derivative of the Ricci and Weyl tensors are of algebraic type {\bf II/D} on the horizon using the spinor formalism. The b.w. of the components of the first covariant derivative of the Ricci and Weyl spinors for non-vacuum type {\bf D} spacetimes are:  

\small
\begin{align} \begin{array}{c|c|c|c|c}
 -2 & -1 & 0 & +1 & +2 \\ \hline 
 & D \Psi_{20'} & D \Psi_{21'},~ D \Psi_{30'} & D_{31'} & \\ [0.5 mm] \hline
 D \Phi_{10'00'}  & D \Phi_{11'00'},~D \Phi_{10'01'} & D \Phi_{11'01'},~ D \Phi_{11'10'} & D \Phi_{11'11'},~ D \Phi_{12'01'} & D \Phi_{12'11'} \\   
 D \Phi_{01'00'} & D \Phi_{10'10'},~D \Phi_{01'01'} & D \Phi_{10'11'},~ D \Phi_{12'00'} & D \Phi_{12'10},~ D \Phi_{21'01'} & D \Phi_{21'11'} \\ 
 &  D \Phi_{01'10'} & D \Phi_{01'11'},~ D \Phi_{21'00'} & D \Phi_{21'10'} & \end{array} \nonumber
\end{align}
\normalsize

\subsection{Kerr-Newman-NUT-(Anti) de Sitter metric}

The 4D Kerr-Newman-NUT-(Anti)-de Sitter metric is given by \cite{PlebDem1976, Griffiths2005, Griffiths2007}:
\beq ds^2 &=& -\frac{Q}{\tilde{\rho}^2} \left[ dt - \left(a\sin^2 \theta + 4 l \sin^2 \frac{\theta}{2}\right) d\phi \right]^2 \label{KNNAds}\\ 
&& +  \frac{\tilde{\rho}^2}{Q}dr^2 + \frac{P}{\tilde{\rho}^2}\left[ a dt - \left(r^2 +(a+l)^2\right) d\phi \right]^2  + \frac{\tilde{\rho}^2}{P} \sin^2 \theta d\theta^2, \nonumber \eeq
where $\tilde{\rho}^2 \equiv \tilde{\rho}(r, \theta)$, $P\equiv P(\theta)$ and $Q \equiv Q(r)$ are functions of $\cos \theta$ and $r$ containing the parameters   $m,e,g, a, l,$ and $\Lambda$ which are, respectively, the mass, the electric and magnetic charges, a rotation parameter, a NUT parameter in a de Sitter or anti-de Sitter background and the cosmological constant, and where: 

\beq
\tilde{\rho}^2 &=& r^2 + (l + a\cos \theta)^2 \label{KNA1}\\ 
P &=& \sin^2 \theta( 1 + \frac43 \Lambda a l \cos \theta + \frac13 \Lambda a^2 \cos^2 \theta) \\
Q &=& (a^2 - l^2 + e^2 + g^2) - 2mr + r^2 \label{KNA3}\\ &&  - \Lambda[(a^2-l^2)l^2 + (\frac13 a^2 + 2l^2)r^2 + \frac13 r^4]. \nonumber \eeq

The locations of the event horizon for this solution are denoted by the roots of $Q(r)$. The expressions for the discriminant SPIs are very large polynomials in $\cos \theta$ and $r$, and it is not clear if they can be factored into irreducible polynomials. Cartan invariants on the other hand allow for the construction of simpler candidates for detecting the horizon.

Defining: 
\beq & t^0 = \frac{\sqrt{Q}}{\rho} \left[ dt - \left(a\sin^2 \theta + 4 l \sin^2 \frac{\theta}{2}\right) d\phi \right],~~t^1 = \frac{\rho}{\sqrt{Q}}dr, & \nonumber \\
&~~t^2 = \frac{\sqrt{P}}{\rho}\left[ a dt - \left(r^2 +(a+l)^2\right) d\phi \right],~~t^3 = \frac{\rho}{\sqrt{P}} \sin \theta d\theta, \nonumber \eeq
the null frame we will work in is:
\begin{eqnarray}
\ell = \frac{t^0 - t^1}{\sqrt{2}},~~ n = \frac{t^0+t^1}{\sqrt{2}},~~m = \frac{t^2-i t^3}{\sqrt{2}},~~\bar{m} = \frac{t^2+i t^3}{\sqrt{2}}. \eeq
The only non-zero NP curvature scalars are $\Lambda_{NP} = \frac16 \Lambda$, 
\beq
 \Psi_{2} &=& -\left( m + i \left( \frac13(a^2-4l^2)l\Lambda \right) \right) \left( \frac{1}{ir +l + a \cos \theta} \right)^3 \nonumber \\  &&+ (e^2 + g^2) \left( \frac{1}{ir +l + a \cos \theta} \right)^3 \left( \frac{1}{-ir +l + a \cos \theta} \right), \label{A6} \eeq
\noindent and
\beq & \Phi_{11} =  \frac12 \frac{e^2+g^2}{| a \cos \theta + l + ir |^2 }. & \label{A7}
\eeq
At zeroth order of the Cartan algorithm we obtain as our Cartan invariants the real and imaginary parts of $\Psi_2$, which are functionally independent ($t_0 =2$). The zeroth order isotropy group consists of boosts and spins (dim $H_0$ = 2).

At the first iteration of the algorithm, we get that the non-zero components of the  covariant derivative of $\Psi$ are: 

\beq D^1 \Psi_{20'} = D \Psi_2 , \quad 
D^1 \Psi_{30'} = 3 \tau \Psi_2 , \quad 
D^1 \Psi_{21'} = \delta \Psi_2 , \quad
D^1 \Psi_{31'} = -3 \rho \Psi_2,  \eeq

\noindent where we have used $\rho = \mu$ and $\pi = -\tau$. The components of the covariant derivative of $\Phi$ are:

\beq & D^1\Phi_{11'00'} = D \Phi_{11},~~ D^1 \Phi_{11'11'} = \Delta \Phi_{11}, & \nonumber \\ 
& D^1 \Phi_{11'01'} = \overline{D^1 \Phi_{11'01'}} = \delta \Phi_{11},& \nonumber \\
& D^1 \Phi_{10'01'} = D^1 \Phi_{12'10'} = \overline{D^1 \Phi_{01'01'}} = \overline{D^1 \Phi_{21'10'}} = \bar{\rho} \Phi_{11},  & \nonumber \\
& D^1 \Phi_{10'11'} = D^1 \Phi_{12'00'} = \overline{ D^1 \Phi_{01'11'}} = \overline{ D^1 \Phi_{21'00'}} = \bar{\tau} \Phi_{11}. & \nonumber \eeq

\noindent These components may be simplified using the Bianchi identities:

\beq &D \Psi_2 = 3 \rho \Psi_2 + 2 \rho \Phi_{11},~~ \Delta \Psi_2 = - 3 \mu \Psi_2 - 2 \mu \Phi_{11}, & \nonumber \\ 
& \delta \Psi_2 = -3\tau \Psi_2 + 2 \tau \Phi_{11},~~\bar{\delta} \Psi_2 = 3\pi \Psi_2 - 2 \pi \Phi_{11}. & \nonumber \\
& D\Phi_{11} = 2(\rho + \bar{\rho}) \Phi_{11},~~ \delta \Phi_{11} = 2(\tau - \bar{\pi}) \Phi_{11},~~ \Delta \Phi_{11} = 2(\mu + \bar{\mu}) \Phi_{11}.& \nonumber \eeq

\noindent Computing the spin-coefficient $\rho = \mu$ we obtain:

\beq \rho = \mu = - \frac12  \frac{\sqrt{Q} [r + i(a\cos\theta + l)]}{\tilde{\rho} | a\cos\theta + l + i r|^2 }.  \eeq

%\normalsize
%
\noindent The location of the event horizons are obtained from the roots of $Q(r)$. It is clear that consequently all b.w. -1 and +1 terms vanish on the horizon. This implies that the first covariant derivative of the Weyl and Ricci tensors are of type {\bf D}. 

Using the formulae $(4.3a')-(4.3i')$ in \cite{ref1,Stewart} we can calculate the second covariant derivatives of the Weyl and Ricci spinors. It may be shown that on the horizon all positive and negative b.w. terms vanish, implying that the second covariant derivative of the Ricci and Weyl spinors are of type {\bf D} as well.

We note that it is possible to apply Theorem \ref{PSthrm} or Theorem \ref{DPSthrm} to construct a SPI that will detect the horizon regardless of the frame \cite{GANG}. For the Kerr-Newmann-NUT-(Anti)-de Sitter metric the cohomogeneity is two, and so we can combine the zeroth order invariants $I_1 = C_{abcd} C^{abcd}$ and $I_2 = C^{*abcd} C_{abcd}$ to produce $\Psi_2$ and $\bar{\Psi}_2$, and compute the norm of ${\bf W} =|| d\Psi_2 \wedge d \bar{\Psi}_2 ||^2$ which is a degree eight, first order SPI that detects the horizon (see below):  
\vspace{ 3 mm}

\beq \frac{9 | a\cos\theta + l + i r|^6}{2} {\bf W} =   - Q (a^2 \sin^2 \theta ) (\Lambda( a^2 \cos^2 \theta + 4 a l \cos \theta) + 3) |\Psi_2|^2 \label{A10} .  \eeq

\newpage 

\section{Appendix B: Relationship between the Page-Shoom invariants and SPIs for the Kerr-Newman-NUT-(Anti)-de Sitter Spacetime}

The Page-Shoom invariants are defined in Theorem \ref{PSthrm} as the norm of the wedge product of the exterior derivatives of $n$ SPIs. We will show that the Page-Shoom invariants and the discriminants for the covariant derivative of the Weyl tensor, $C_{abcd;e}$, share a common zero in the stationary Kerr-Newman-NUT-(Anti)-de Sitter black hole solution. Due to the extensive size of these invariants relative to a coordinate system, we will use the NP formalism to show this property relative to the frame given in \cite{GANG}.   

The 4D Kerr-Newman-NUT-(Anti)-de Sitter metric is given by the line element \eqref{KNNAds} with metric functions defined in (\ref{KNA1}-\ref{KNA3}) containing the parameters $m,e,g, a, l,$ and $\Lambda$ which are, respectively, mass, the electric and magnetic charges, the rotation parameter, the NUT parameter in a de Sitter or anti-de Sitter background, and the cosmological constant. The surfaces for which $Q(r)=0$ denote the horizons of this solution \cite{Griffiths2005}.

% \cite{PlebDem1976, Griffiths2005, Griffiths2007}
%\beq ds^2 &=& -\frac{Q}{\tilde{\rho}^2} \left[ dt - \left(a\sin^2 \theta + 4 l \sin^2 \frac{\theta}{2}\right) d\phi \right]^2 + \frac{\tilde{\rho}^2}{Q}dr^2 \\ 
%&& + \frac{P}{\tilde{\rho}^2}\left[ a dt - \left(r^2 +(a+l)^2\right) d\phi \right]^2  + \frac{\tilde{\rho}^2}{P} \sin^2 \theta d\theta^2 \nonumber \eeq
%where $\tilde{\rho} \equiv \tilde{\rho}(r, \theta)$, $P\equiv P(\theta)$ and $Q \equiv Q(r)$ are functions of $\cos \theta$ and $r$, containing the parameters   $m,e,g, a, l,$ and $\Lambda$ which are, respectively, mass, the electric and magnetic charges, the rotation parameter, the NUT parameter in a de Sitter or anti-de Sitter background, and the cosmological constant: 
%
%\begin{eqnarray}
%\tilde{\rho}^2 &=& r^2 + (l + a\cos \theta)^2 \nonumber \\ 
%P &=& \sin^2 \theta( 1 + \frac43 \Lambda a l \cos \theta + \frac13 \Lambda a^2 \cos^2 \theta) \nonumber \\
%Q &=& (a^2 - l^2 + e^2 + g^2) - 2mr + r^2 - \Lambda[(a^2-l^2)l^2 + (\frac13 a^2 + 2l^2)r^2 + \frac13 r^4] \nonumber \end{eqnarray}
%
%\noindent For this solution the surfaces for which $Q(r)=0$ are horizons \cite{Griffiths2005}.

As the cohomogeneity of the Kerr-Newman-NUT-(Anti)-de Sitter spacetime is two, any Page-Shoom invariant is then the norm of the wedge product of the exterior derivatives of two SPIs. We will work with the simplest Page-Shoom invariant, constructed from linear combinations of the zeroth order SPIs defined in terms of $I_1$ and $I_2$ in \eqref{I1} and \eqref{I2}:

\beq \Psi_2^2 = \frac{( I_1 - I_2)}{48},~~\bar{\Psi}_2^2 = \frac{( I_1 + I_2)}{48}, \nonumber \eeq

\noindent where $\Psi_2$ is the only non-zero NP Weyl scalar. From the perspective of the Cartan algorithm, it is sufficient to study the roots of the resulting Page-Shoom invariant, as any other SPI can be expressed in terms of the functionally independent Cartan invariants $\Psi_2$ and its complex conjugate. For any two functionally independent SPIs, the Page-Shoom invariant is: 
\beq {\bf W}' = ||d J_1 \wedge d J_2||^2,~~ J_i = F_i(\Psi_2, \bar{\Psi}_2),~i=1,2. \nonumber \eeq 
\noindent Applying the chain rule will ensure that this can be expressed in terms of the frame derivatives of $\Psi_2$ and $\bar{\Psi}_2$:

\beq {\bf W}' \propto ||d \Psi_2 \wedge d \bar{\Psi}_2||^2 = {\bf W}. \nonumber \eeq

Constructing the wedge product, $d\Psi_2 \wedge d \bar{\Psi}_2$ and using the Bianchi identities \cite{GANG}, the norm is then:

\beq {\bf W} &=& - 16 |3\mu \Psi_2 + 2\mu \Phi_{11}|^2 | 3  \pi \Psi_2 - 2 \pi \Phi_{11}|^2 + \nonumber \\
&&  16 \mathfrak{R} [ (\mu \Psi_2 + 2 \mu \Phi_{11})^2 (3 \bar{\pi} \bar{\Psi}_2 - 2 \bar{\pi} \Phi_{11})^2 ], \label{WeqnSym} \eeq

\noindent where the frame dependent identity $\tau = -\pi$ has been used and $\Phi_{11}$ is the only non-zero NP Ricci scalar. Relative to the coordinate system, $\Psi_2$ and $\Phi_{11}$ are given by equations \eqref{A6} and \eqref{A7}, 
% then:
%
%
%\beq \Psi_{2} &=& -\left( m + i \left( \frac13(a^2-4l^2)l\Lambda \right) \right) \left( \frac{1}{ir +l + a \cos \theta} \right)^3 \nonumber \\ &&+ (e^2 + g^2) \left( \frac{1}{ir +l + a \cos \theta} \right)^3 \left( \frac{1}{-ir +l + a \cos \theta} \right), \nonumber \\ 
%\Phi_{11} &=& \frac{e^2 + g^2}{2(a^2 \cos^2 \theta +2 a l \cos \theta + l^2 + r^2)^2},\eeq
%
%
\noindent and the relevant spin-coefficients are: 

\begin{eqnarray}
&\mu = \rho = - \frac12  \frac{\sqrt{Q} [r + i(a\cos\theta + l)]}{\tilde{\rho} | a\cos\theta + l + i r|^2 }, & \nonumber \\ 
& \tau = -\pi = -\frac{1}{\sqrt{2}} \frac{a (ir - a \cos \theta - l) \sqrt{1 + \frac43 \Lambda a l \cos \theta + \frac13 \Lambda a^2 \cos^2 \theta} }{\tilde{\rho} (\cos^2 \theta a^2 + 2 a l \cos \theta + l^2 + r^2)}. & \nonumber \end{eqnarray}

\noindent It is clear that $\rho =\mu$ vanishes on the stationary horizons, and so the invariant ${\bf W}$ must vanish there as well. Indeed, from equation \eqref{A10} the invariant may be expressed as, 
\beq {\bf W} = Q F(r, \cos \theta), \label{WeqnKerr} \eeq
\noindent where $F$ is a ratio of two polynomials in terms of $r$ and $\cos \theta$, arising from the coordinate form of the related spin coefficient $\mu$ and $\pi$, and the curvature scalars \eqref{WeqnSym}.

The discriminants related to the covariant derivative of the Weyl tensor will also detect the horizon since $C_{abcd;e}$ will be of algebraic type {\bf II/D} on the horizon. We will examine the first discriminant, ${^1}D$, given in equation \eqref{rictypeii212}, but the result carries over to the second discriminant, ${^2}D$, whose form is similar to \eqref{rictypeii212} with invariants defined in \eqref{riccidef1}. Both of these invariants  act as necessary conditions for $C_{abcd;e}$ to be of type {\bf II/D}. 

The first discriminant is constructed  using the trace-free symmetric tensor ${^1}S^{a}_{~b}={^1}T^{a}_{~b}$ defined in \eqref{T1a} and \eqref{T1b}. 
%\begin{equation}
%{^1}T^{a}_{~b} \equiv C^{cdef;a}C_{cdef;b} - \frac{1}{4} \delta^{a}_{~b} {^1}I_2,
%\end{equation}
The necessary condition for this operator to be of type {\bf II/D} is \eqref{rictypeii212}:
\begin{equation}
{^1} D =  -{^1}s_3^2(4 {^1}s_2^3 - 144 {^1}s_2 {^1}s_4 + 27{^1}s_3^2) + {^1}s_4(16 {^1}s_2 - 128 {^1}s_4 {^1}s_2^2  + 256 {^1}s_4^2) = 0  \end{equation} 
\noindent where ${^1}s_2, {^1}s_3,$ and ${^1}s_4$ are defined in \eqref{riccidef12}. It may be shown by direct computation that ${^1}D$ shares a common root with ${\bf W}$: 
\beq {^1} D =  Q^6 G(r, \cos \theta), \nonumber \eeq
\noindent where $G$ is a ratio of polynomial functions. 

The fact that ${\bf W}$ and ${^1}{\bf D}$ share roots is in itself unsurprising as they both vanish on the horizon. Through direct knowledge of the frame in which the covariant derivative of the Weyl tensor, $C_{abcd;e}$ is of type {\bf II/D} on the horizon \cite{GANG}, we have shown for the Kerr-Newman-NUT-(Anti)-de Sitter spacetime the invariant ${\bf W}$ and the discriminant for the Weyl curvature tensor, ${\bf {^1} D}$, vanish when $ C_{abcd;e}$ is of type {\bf II}/{\bf D}. Therefore, for the Kerr-Newman-NUT-(Anti)-de Sitter solution, the invariant {\bf W} provides an SPI of lower order than ${^1}D$ giving necessary conditions for when the $C_{abcd;e}$  is of type {\bf II/D}. 

%Want to relate to Page-Shoom invariants.
%Use stationarity to simplify syzygys (for Kerr).
%Exploit stationarity: add in text (and refer to appendix A).
%If there exists any other geometric (invariantly defined)
%structure, such as for example an invariantly defined timelike vector
%field, we can construct other (lower order invariants).
%Ex: perfect fluid solutions with a timelike Ricci eigenfunction and a timelike Killing vector (stationary spacetime): see
%{[Appendix E: scalar invariants in stationary spacetimes and type D conditions]}.
%{\bf{D: We mention this in section 3 and the discussion as a possible project (second paragraph) }}

\newpage

%xxyyzz What the fuck do you mean by 4D8, page 41 of edits.

\section{Appendix C: Geometric Identities} \label{geomid}

\subsection{FKWC Bases in 4D}

\gls{fkwc} \cite{Fulling} systematically expanded the Riemann polynomials encountered in calculations on standard bases constructed from group theoretical considerations. Bases for scalar Riemann polynomials of order eight or less in the derivatives of the metric tensor and for tensorial Riemann polynomials of order six or less were presented.  The FKWC-bases were modified to be dimensionally independent \cite{DecaniniFolacci2007}, allowing for irreducible expressions \cite{CH2}.

The geometrical identities utilized to
eliminate ``spurious" Riemann monomials of order up to six by
expressing them in terms of elements of FKWC-bases have been derived using: 
\begin{enumerate}
\item The commutation of covariant derivatives in the form
\begin{eqnarray}\label{CD_NabNabTensor}
&  & T^{a \dots}_{~~~~ b \dots ;c d} -
T^{a \dots}_{~~~~ b \dots ;d c} = \nonumber \\
&   & \qquad  + R^{a}_{~ e d c} T^{e
\dots}_{~~~~ b \dots } + \dots -
R^{e}_{~ b d c} T^{a
\dots}_{~~~~ e \dots} - \dots.
\end{eqnarray}
\item The ``symmetry" properties of the Ricci and the
Riemann tensors (pair symmetry, antisymmetry, cyclic symmetry)

\begin{eqnarray}
& & R_{ab}=R_{ba},  \label{SymRicci} \\
& & R_{abcd}=R_{cdab},  \label{SymRiemann_1} \\
& & R_{abcd}=-R_{bacd} \quad \mathrm{and} \quad R_{abcd}=-R_{abdc},
\label{SymRiemann_2} \\
& & R_{abcd}+R_{adbc}+R_{acdb}=0.  \label{SymRiemann_Cycl}
\end{eqnarray}
\item The differential Bianchi identity and the identities obtained by contraction of index pairs

\begin{eqnarray}
& & R_{abcd;e}+R_{abec;d}+R_{abde;c}=0, \label{AppBianchi_1} \\
& & R_{ abcd}^{~~~~~;d}= R_{ac;b}-R_{bc;a},  \label{AppBianchi_2a} \\
& &R_{ ab}^{~~;b}= (1/2)  R_{;a}. \label{AppBianchi_2b}
\end{eqnarray}
\end{enumerate}

\subsubsection{Geometric identities}

To refer to equations in \cite{DecaniniFolacci2007}, we will denote equation numbers as [DF 1] etc. Using the ``symmetry" properties of the Ricci and the Riemann tensors (\ref{SymRicci})-(\ref{SymRiemann_Cycl}), the Bianchi identity and its consequences (\ref{AppBianchi_1})-(\ref{AppBianchi_2b}), the commutation of covariant derivatives  (\ref{CD_NabNabTensor}), along with several geometric zeroth order identities for the curvature tensor and Ricci tensor (which will not be displayed), the following geometric identity is obtained [DF 7]:
\begin{eqnarray}\label{Simpl_R0_O6_1}
&  R_{pqrs} \Box R^{pqrs} = 4 \, R_{pq ; rs}R^{prqs}
+2 \, R_{pq}R^p_{~ rst} R^{qrst } \nonumber \\
&  \qquad - \, R_{pqrs}R^{pq}_{~~ uv} R^{rsuv} - 4\,
R_{prqs} R^{p ~ q}_{~ u ~ v} R^{r u s
v}.
\end{eqnarray}

Particular results can be written in terms of the Weyl tensor by substituting out trace parts and working with identities for the Ricci tensor. Or we can consider the vacuum case with vanishing Ricci tensor in applications, which leads to an extra constraint on $C_{ab[cd;e]}$ which will simplify the SPIs.

For the SPIs involving the first and second covariant derivatives of the Riemann tensor we may employ the geometrical identities [DF 8-10, 13-15, 19, 20-29, 32, 33-34],
including:

\begin{equation}\label{Simpl_R1_O5_3}
 R^{pqrs}R_{pqra;s}=\frac{1}{2} \, R^{pqrs}R_{pqrs;a},
\end{equation}

\begin{equation}\label{Simpl_R2_O6_6}
R^{pqrs} R_{pqr a;sb}= \frac{1}{2} \, R^{pqrs} R_{pqrs ; a b},   
\end{equation}

\begin{eqnarray}\label{Simpl_R3_O5_1}
& R^{pqr}_{~~~ a}R_{rqpb;c} =\frac{1}{2} \,
R^{pqr}_{~~~ a} R_{pqrb;c}  ,
\end{eqnarray}
and the geometrical identities [DF 38-39, 40-47] and [DF 48-55] including, in particular:

\begin{eqnarray}\label{Simpl_R4_O6_NEW1}
& &  \Box R_{abcd} =  R_{ac;bd} - R_{bc;ad} -R_{ad;bc} + R_{bd;ac} \nonumber \\
& & \qquad + R^p_{~ c}R_{pdab} - R^p_{~
d}R_{pcab}\nonumber \\
& & \qquad +2\, R^{p ~ q}_{~ a ~
d}R_{pbqc} - 2\, R^{p ~ q}_{~ a ~
c}R_{pbqd} \nonumber \\
& & \qquad - 2\, R^{p ~ q}_{~ a ~
b}R_{pcqd} + 2\, R^{p ~ q}_{~ a ~
b}R_{pdqc}.
\end{eqnarray}

\noindent The geometrical identity (\ref{Simpl_R4_O6_NEW1}) in vacuum becomes:

\begin{eqnarray}\label{Simpl_R4_O6_NEW1VAC}
& &  \Box C_{abcd} =  2\, C^{p ~ q}_{~ a ~d}C_{pbqc} - 2\, C^{p ~ q}_{~ a ~c}C_{pbqd} \nonumber \\
& & \qquad - 2\, C^{p ~ q}_{~ a ~b}C_{pcqd} + 2\, C^{p ~ q}_{~ a ~ b}C_{pdqc},
\end{eqnarray}
whence (\ref{Simpl_R0_O6_1}) reduces to (in vacuum):
\begin{eqnarray}\label{Simpl_R0_O6_1VAC}
&  C_{pqrs} \Box C^{pqrs} =  - \, C_{pqrs}C^{pq}_{~~ uv} C^{rsuv} - 4\,
C_{prqs} C^{p ~ q}_{~ u ~ v} C^{r u s v}.
\end{eqnarray}

\noindent For black holes with a vacuum exterior, these identities may help simplify the related discriminants.

\subsubsection{Conserved tensor quantities}

In addition, the results of \cite{DecaniniFolacci2007} provide irreducible expressions for the metric variations (i.e., for  the functional derivatives with respect to the metric tensor) of the action terms associated with the 17 basis elements for the so-called FKWC curvature invariants of order six. For every scalar $S$ that appears in the action (gravitational Lagrangian) we obtain by variation (since these geometric tensors 
are automatically conserved due to the invariance of the actions under spacetime diffeomorphisms) a symmetric conserved rank-2 tensor, e.g., $S_{ab}$  with $S_{ab}^{~~;b}=0$ (from which syzygys can be obtained, such as $S_{ab}^{~~;b} S^{a~~;c}_{~c} =0$). As above, to refer to the $n^{th}$ equation in \cite{CH2} we write [CH n].

Choosing $S$ to be the Ricci scalar, $R$, we obtain

\begin{equation}\label{einstein} 
S_{ab}=R_{ab} -  \frac{1}{2}Rg_{ab},  
\end{equation}
yielding the usual contracted Bianchi identity.
Choosing $S=R_{pqrs}R^{pquv} R^{rs}_{\phantom{rs} uv}$ we obtain [CH 6]:
\begin{eqnarray}\label{FD_06_9} 
S_{ab} &= 24\,  R^{p \phantom{(a}; qr}_{\phantom{p } (a} R_{|pqr|  b)} 
-12\, R^p_{\phantom{p} a;q} R_{p b}^{\phantom{p b};q} + 12\, 
R^p_{\phantom{p} a;q} R^q_{\phantom{q} b;p}    \nonumber \\ 
& \quad + 3\,  
R^{pqrs}_{\phantom{pqrs};a} R_{pqrs ; b }  -6\, 
R^{pqr}_{\phantom{pqr}a;s}R_{pqr b}^{\phantom{pqr b};s} 
-6\, R^{pq}R^{rs}_{\phantom{rs} pa}R_{rs q b}   \nonumber \\ 
& \quad +12\, R^{p r q 
s}R^t_{\phantom{t} pq a}R_{t rs b}  
+ \frac{1}{2}g_{ab} [R_{pqrs}R^{pquv} R^{rs}_{\phantom{rs}uv} ]. 
\end{eqnarray} 
In vacuum this becomes:

\begin{eqnarray}\label{FD_06_9a} 
S_{ab}  = 3\, C^{pqrs}_{\phantom{pqrs};a} C_{pqrs ; b }   
-6\,C^{pqr}_{\phantom{pqr}a;s}C_{pqr b}^{\phantom{pqr b};s} +12\, C^{p r q 
s}C^t_{\phantom{t} pq a}C_{t rs b}   \nonumber \\ 
+\frac{1}{2} g_{ab} [C_{pqrs}C^{pquv} C^{rs}_{\phantom{rs} 
uv}].
\end{eqnarray}

Also, in vacuum, by varying  
$S= C_{pqrs} C^{pqrs}$ we obtain a  
symmetric conserved rank-2 tensor which depends on quadratic polynomial 
contractions of the Weyl tensor [CH 3]:
\begin{equation}\label{C1} 
S_{ab}=C_{almn}C_{b}^{~lmn}+ C_{blmn}C_{a}^{~lmn}. 
\end{equation}.

\subsection{Syzygies on the Horizon}
If an invariant $I$ vanishes on a surface, then the pullback of $\nabla I$ will also vanish on the surface. In principle, we can apply this to construct syzygies using the covariant derivatives of the Ricci and Weyl tensors that hold on the horizon. For example, suppose that a necessary algebraic condition is of the form $F(^{1}I, ^{2}I)=0$. Then by differentiation (covariant differentiation is simply partial differentiation here), we obtain $(F_1) ^{1}I_{,a} + (F_2) ^{2}I_{,a} = 0$ (where $F_i$ denotes differentiation with respect to $^{i}I$),  whence by contraction with  $^{1}I_{,a},^{2}I_{,a}$ we obtain scalar identities; in particular, by addition, we obtain:
\begin{equation}
(F_1^2) {^{1}}I_{,a} {^{1}}I^{,a}  =   (F_2^2) {^{2}}I_{,a}{^{2}}I^{,a}.
\end{equation}

\newpage

\section{Appendix D: Necessary Type {\bf II/D} Conditions for the Weyl Tensor Using Discriminants} \label{AppendixOuch}

Treating the Weyl tensor as a trace-free operator on the vector space of bivectors, we can apply the discriminant
analysis to produce necessary conditions for the operator to be of type {\bf II/D}, using the following traces of the powers
of the Weyl operator (in terms of the $W_i$ defined earlier): 
\beq\tilde{W}_2 = 8 W_2,
\tilde{W}_3  = 16 W_3, \tilde{W}_4 = 32 W_4, \tilde{W}_6 = 128 W_6, \nonumber \eeq
\noindent  and 
 \beq \tilde{W}_5 \equiv C_{abcd}C^{cd}_{~~pq}C^{pq}_{~~r s}C^{rs}_{~~tu}C^{tuab}. \nonumber \eeq
 
\noindent The necessary conditions that the Weyl operator is of type {\bf II/D} are ${^5} D_6 = {^6}D_5 =0$, where these are given by the following equations: 
\beq {^6}D_5 &=& 6\,{{ \tilde{W}_5}}^{4}-\frac12 \,{{ \tilde{W}_4}}^{5}+{\frac {251}{60}}\,{{ \tilde{W}_2}}^{4
}{ \tilde{W}_3}\,{ \tilde{W}_4}\,{ \tilde{W}_5}-{\frac {41}{10}}\,{{ \tilde{W}_2}}^{3}{ \tilde{W}_3}
\,{ \tilde{W}_5}\,{ \tilde{W}_6}+\frac12\,{{ \tilde{W}_2}}^{2}{{ \tilde{W}_3}}^{2}{ \tilde{W}_4}\,{\it 
\tilde{W}_6} \nonumber \\ 
&&-{\frac {99}{10}}\,{{ \tilde{W}_2}}^{2}{ \tilde{W}_3}\,{{ \tilde{W}_4}}^{2}{ \tilde{W}_5}+\frac25\,{ \tilde{W}_2}\,{{ \tilde{W}_3}}^{3}{ \tilde{W}_4}\,{ \tilde{W}_5}+{\frac {49}{576}}\,{{
 \tilde{W}_2}}^{7}{{ \tilde{W}_3}}^{2}-\frac{1}{48}\,{{ \tilde{W}_2}}^{7}{ \tilde{W}_6}+\frac{1}{48}\,{{ \tilde{W}_2
}}^{6}{{ \tilde{W}_4}}^{2}\nonumber \\ 
&&+{\frac {17}{32}}\,{{ \tilde{W}_2}}^{4}{{ \tilde{W}_3}}^{4}+{
\frac {73}{200}}\,{{ \tilde{W}_2}}^{5}{{ \tilde{W}_5}}^{2}-\frac38 \,{{ \tilde{W}_2}}^{4}{{
 \tilde{W}_4}}^{3}-\frac29 \,{ \tilde{W}_2}\,{{ \tilde{W}_3}}^{6}-\frac23 \,{{ \tilde{W}_2}}^{4}{{ \tilde{W}_6
}}^{2}+\frac74 \,{{ \tilde{W}_2}}^{2}{{ \tilde{W}_4}}^{4}\nonumber \\ 
&&+{\frac {8}{15}}\,{{ \tilde{W}_3}}^{
5}{ \tilde{W}_5}+{\frac {9}{8}}\,{{ \tilde{W}_3}}^{4}{{ \tilde{W}_4}}^{2}-6\,{ \tilde{W}_2}\,{
{ \tilde{W}_6}}^{3}+6\,{{ \tilde{W}_4}}^{2}{{ \tilde{W}_6}}^{2}-{\frac {11}{30}}\,{{\it 
\tilde{W}_2}}^{6}{ \tilde{W}_3}\,{ \tilde{W}_5}-{\frac {263}{288}}\,{{ \tilde{W}_2}}^{5}{{ \tilde{W}_3}
}^{2}{ \tilde{W}_4}\nonumber \\ 
&&+\frac38 \,{{ \tilde{W}_2}}^{5}{ \tilde{W}_4}\,{ \tilde{W}_6}+{\frac {49}{72}}\,
{{ \tilde{W}_2}}^{4}{{ \tilde{W}_3}}^{2}{ \tilde{W}_6}-{\frac {313}{90}}\,{{ \tilde{W}_2}}^{3}
{{ \tilde{W}_3}}^{3}{ \tilde{W}_5}+{\frac {251}{144}}\,{{ \tilde{W}_2}}^{3}{{ \tilde{W}_3}}^{2
}{{ \tilde{W}_4}}^{2}-{\frac {25}{24}}\,{{ \tilde{W}_2}}^{2}{{ \tilde{W}_3}}^{4}{ \tilde{W}_4}
\nonumber \\ 
&&-{\frac {13}{12}}\,{{ \tilde{W}_2}}^{3}{{ \tilde{W}_4}}^{2}{ \tilde{W}_6}-{\frac {106}{
25}}\,{{ \tilde{W}_2}}^{3}{ \tilde{W}_4}\,{{ \tilde{W}_5}}^{2}+{\frac {256}{25}}\,{{\it 
\tilde{W}_2}}^{2}{{ \tilde{W}_3}}^{2}{{ \tilde{W}_5}}^{2}+2\,{ \tilde{W}_2}\,{{ \tilde{W}_3}}^{4}{\it 
\tilde{W}_6}+{\frac {7}{8}}\,{ \tilde{W}_2}\,{{ \tilde{W}_3}}^{2}{{ \tilde{W}_4}}^{3}\nonumber \\ 
&&+6\,{{ \tilde{W}_2
}}^{2}{ \tilde{W}_4}\,{{ \tilde{W}_6}}^{2}+{\frac {27}{5}}\,{{ \tilde{W}_2}}^{2}{{ \tilde{W}_5
}}^{2}{ \tilde{W}_6}-3\,{ \tilde{W}_2}\,{{ \tilde{W}_3}}^{2}{{ \tilde{W}_6}}^{2}-{\frac {68}{5
}}\,{ \tilde{W}_2}\,{ \tilde{W}_3}\,{{ \tilde{W}_5}}^{3}-\frac{11}{2} \,{ \tilde{W}_2}\,{{ \tilde{W}_4}}^{3}{
 \tilde{W}_6}\nonumber \\ 
&&+{\frac {99}{10}}\,{ \tilde{W}_2}\,{{ \tilde{W}_4}}^{2}{{ \tilde{W}_5}}^{2}-{
\frac {26}{5}}\,{{ \tilde{W}_3}}^{3}{ \tilde{W}_5}\,{ \tilde{W}_6}-\frac{11}{2} \,{{ \tilde{W}_3}}^{2}{
{ \tilde{W}_4}}^{2}{ \tilde{W}_6}+{\frac {12}{5}}\,{{ \tilde{W}_3}}^{2}{ \tilde{W}_4}\,{{\it 
\tilde{W}_5}}^{2}+{\frac {14}{5}}\,{ \tilde{W}_3}\,{{ \tilde{W}_4}}^{3}{ \tilde{W}_5}\nonumber \\ 
&&+12\,{ \tilde{W}_3
}\,{ \tilde{W}_5}\,{{ \tilde{W}_6}}^{2}-18\,{ \tilde{W}_4}\,{{ \tilde{W}_5}}^{2}{ \tilde{W}_6}+{
\frac {39}{5}}\,{ \tilde{W}_2}\,{ \tilde{W}_3}\,{ \tilde{W}_4}\,{ \tilde{W}_5}\,{ \tilde{W}_6}, \nonumber 
\eeq

%======

\beq  
&{^6}D_6 = {\frac {221}{36}}\,{{ \tilde{W}_2}}^{4}{{ \tilde{W}_3}}^{2}{ \tilde{W}_4}\,{{ \tilde{W}_6}}^{2
}+{\frac {4141}{3600}}\,{{ \tilde{W}_2}}^{4}{{ \tilde{W}_3}}^{2}{{ \tilde{W}_5}}^{2}{
 \tilde{W}_6}+{\frac {3529}{3000}}\,{{ \tilde{W}_2}}^{4}{ \tilde{W}_3}\,{ \tilde{W}_4}\,{{\it 
\tilde{W}_5}}^{3}+{\frac {3361}{432}}\,{{ \tilde{W}_2}}^{3}{{ \tilde{W}_3}}^{2}{{ \tilde{W}_4}}^{
3}{ \tilde{W}_6}& \nonumber \\ 
&+{\frac {10379}{7200}}\,{{ \tilde{W}_2}}^{3}{{ \tilde{W}_3}}^{2}{{ \tilde{W}_4
}}^{2}{{ \tilde{W}_5}}^{2}+{\frac {19}{240}}\,{{ \tilde{W}_2}}^{3}{ \tilde{W}_3}\,{{\it 
\tilde{W}_4}}^{4}{ \tilde{W}_5}+{\frac {53}{360}}\,{{ \tilde{W}_2}}^{2}{{ \tilde{W}_3}}^{5}{\it 
\tilde{W}_5}\,{ \tilde{W}_6}+{\frac {3119}{576}}\,{{ \tilde{W}_2}}^{2}{{ \tilde{W}_3}}^{4}{{\it 
\tilde{W}_4}}^{2}{ \tilde{W}_6}& \nonumber \\ 
&-{\frac {41}{1200}}\,{{ \tilde{W}_2}}^{2}{{ \tilde{W}_3}}^{4}{\it 
\tilde{W}_4}\,{{ \tilde{W}_5}}^{2}+{\frac {799}{720}}\,{{ \tilde{W}_2}}^{2}{{ \tilde{W}_3}}^{3}{{
 \tilde{W}_4}}^{3}{ \tilde{W}_5}+{\frac {23}{18}}\,{ \tilde{W}_2}\,{{ \tilde{W}_3}}^{6}{ \tilde{W}_4
}\,{ \tilde{W}_6}+{\frac {29}{45}}\,{ \tilde{W}_2}\,{{ \tilde{W}_3}}^{5}{{ \tilde{W}_4}}^{2}{
 \tilde{W}_5}& \nonumber \\ 
&-{\frac {337}{300}}\,{{ \tilde{W}_2}}^{3}{ \tilde{W}_3}\,{{ \tilde{W}_5}}^{3}{
 \tilde{W}_6}-{\frac {637}{600}}\,{{ \tilde{W}_2}}^{3}{{ \tilde{W}_4}}^{2}{{ \tilde{W}_5}}^{2}
{ \tilde{W}_6}-{\frac {133}{60}}\,{{ \tilde{W}_2}}^{2}{{ \tilde{W}_3}}^{3}{ \tilde{W}_5}\,{{
 \tilde{W}_6}}^{2}-{\frac {139}{8}}\,{{ \tilde{W}_2}}^{2}{{ \tilde{W}_3}}^{2}{{ \tilde{W}_4}}^
{2}{{ \tilde{W}_6}}^{2}& \nonumber \\ 
&-{\frac {463}{300}}\,{{ \tilde{W}_2}}^{2}{ \tilde{W}_3}\,{{ \tilde{W}_4
}}^{2}{{ \tilde{W}_5}}^{3}-{\frac {43}{6}}\,{ \tilde{W}_2}\,{{ \tilde{W}_3}}^{4}{ \tilde{W}_4}
\,{{ \tilde{W}_6}}^{2}+{\frac {26}{25}}\,{ \tilde{W}_2}\,{{ \tilde{W}_3}}^{4}{{ \tilde{W}_5}}^
{2}{ \tilde{W}_6}-\frac{1}{25} \,{ \tilde{W}_2}\,{{ \tilde{W}_3}}^{3}{ \tilde{W}_4}\,{{ \tilde{W}_5}}^{3}& \nonumber \\ 
&-{
\frac {25}{16}}\,{ \tilde{W}_2}\,{{ \tilde{W}_3}}^{2}{{ \tilde{W}_4}}^{4}{ \tilde{W}_6}+{
\frac {327}{400}}\,{ \tilde{W}_2}\,{{ \tilde{W}_3}}^{2}{{ \tilde{W}_4}}^{3}{{ \tilde{W}_5}}^{2
}+{\frac {29}{120}}\,{ \tilde{W}_2}\,{ \tilde{W}_3}\,{{ \tilde{W}_4}}^{5}{ \tilde{W}_5}-{
\frac {77}{60}}\,{{ \tilde{W}_3}}^{5}{ \tilde{W}_4}\,{ \tilde{W}_5}\,{ \tilde{W}_6}& \nonumber \\ 
&+{\frac {22
}{5}}\,{{ \tilde{W}_2}}^{2}{ \tilde{W}_3}\,{ \tilde{W}_5}\,{{ \tilde{W}_6}}^{3}+{\frac {21}{5}
}\,{{ \tilde{W}_2}}^{2}{ \tilde{W}_4}\,{{ \tilde{W}_5}}^{2}{{ \tilde{W}_6}}^{2}+{\frac {71}{4}
}\,{ \tilde{W}_2}\,{{ \tilde{W}_3}}^{2}{ \tilde{W}_4}\,{{ \tilde{W}_6}}^{3}-\frac{3}{10} \,{ \tilde{W}_2}\,{{
 \tilde{W}_3}}^{2}{{ \tilde{W}_5}}^{2}{{ \tilde{W}_6}}^{2}& \nonumber \\ 
&-{\frac {79}{20}}\,{ \tilde{W}_2}\,{
{ \tilde{W}_4}}^{3}{{ \tilde{W}_5}}^{2}{ \tilde{W}_6}+{\frac {49}{10}}\,{{ \tilde{W}_3}}^{3}{
 \tilde{W}_4}\,{ \tilde{W}_5}\,{{ \tilde{W}_6}}^{2}-{\frac {31}{20}}\,{{ \tilde{W}_3}}^{2}{{
 \tilde{W}_4}}^{2}{{ \tilde{W}_5}}^{2}{ \tilde{W}_6}-\frac{3}{10} \,{ \tilde{W}_3}\,{{ \tilde{W}_4}}^{4}{\it 
\tilde{W}_5}\,{ \tilde{W}_6}& \nonumber \\ 
&-6\,{ \tilde{W}_3}\,{ \tilde{W}_4}\,{ \tilde{W}_5}\,{{ \tilde{W}_6}}^{3}-{\frac {
163}{5760}}\,{{ \tilde{W}_2}}^{9}{ \tilde{W}_3}\,{ \tilde{W}_4}\,{ \tilde{W}_5}+\frac{1}{16} \,{{ \tilde{W}_2
}}^{8}{ \tilde{W}_3}\,{ \tilde{W}_5}\,{ \tilde{W}_6}+{\frac {527}{864}}\,{{ \tilde{W}_2}}^{7}{
{ \tilde{W}_3}}^{2}{ \tilde{W}_4}\,{ \tilde{W}_6}& \nonumber \\ 
&+{\frac {337}{1440}}\,{{ \tilde{W}_2}}^{7}{
 \tilde{W}_3}\,{{ \tilde{W}_4}}^{2}{ \tilde{W}_5}+{\frac {1127}{1728}}\,{{ \tilde{W}_2}}^{6}{{
 \tilde{W}_3}}^{3}{ \tilde{W}_4}\,{ \tilde{W}_5}-{\frac {3491}{4320}}\,{{ \tilde{W}_2}}^{5}{{
 \tilde{W}_3}}^{3}{ \tilde{W}_5}\,{ \tilde{W}_6}-{\frac {103}{27}}\,{{ \tilde{W}_2}}^{5}{{\it 
\tilde{W}_3}}^{2}{{ \tilde{W}_4}}^{2}{ \tilde{W}_6}& \nonumber \\ 
&-{\frac {15487}{14400}}\,{{ \tilde{W}_2}}^{5}{
{ \tilde{W}_3}}^{2}{ \tilde{W}_4}\,{{ \tilde{W}_5}}^{2}-{\frac {13}{20}}\,{{ \tilde{W}_2}}^{5}
{ \tilde{W}_3}\,{{ \tilde{W}_4}}^{3}{ \tilde{W}_5}-{\frac {7759}{3456}}\,{{ \tilde{W}_2}}^{4}{
{ \tilde{W}_3}}^{4}{ \tilde{W}_4}\,{ \tilde{W}_6}-{\frac {3787}{2160}}\,{{ \tilde{W}_2}}^{4}{{
 \tilde{W}_3}}^{3}{{ \tilde{W}_4}}^{2}{ \tilde{W}_5}& \nonumber \\ 
&-{\frac {289}{720}}\,{{ \tilde{W}_2}}^{3}{
{ \tilde{W}_3}}^{5}{ \tilde{W}_4}\,{ \tilde{W}_5}+{\frac {349}{360}}\,{{ \tilde{W}_2}}^{5}{
 \tilde{W}_3}\,{ \tilde{W}_5}\,{{ \tilde{W}_6}}^{2}+{\frac {569}{1200}}\,{{ \tilde{W}_2}}^{5}{
 \tilde{W}_4}\,{{ \tilde{W}_5}}^{2}{ \tilde{W}_6}-{\frac {751}{720}}\,{{ \tilde{W}_2}}^{6}{
 \tilde{W}_3}\,{ \tilde{W}_4}\,{ \tilde{W}_5}\,{ \tilde{W}_6}& \nonumber \\ 
&+{\frac {337}{72}}\,{{ \tilde{W}_2}}^{4
}{ \tilde{W}_3}\,{{ \tilde{W}_4}}^{2}{ \tilde{W}_5}\,{ \tilde{W}_6}+{\frac {1529}{360}}\,{{
 \tilde{W}_2}}^{3}{{ \tilde{W}_3}}^{3}{ \tilde{W}_4}\,{ \tilde{W}_5}\,{ \tilde{W}_6}-{\frac {253}{30
}}\,{{ \tilde{W}_2}}^{3}{ \tilde{W}_3}\,{ \tilde{W}_4}\,{ \tilde{W}_5}\,{{ \tilde{W}_6}}^{2}& \nonumber \\ 
&-{\frac 
{201}{100}}\,{{ \tilde{W}_2}}^{2}{{ \tilde{W}_3}}^{2}{ \tilde{W}_4}\,{{ \tilde{W}_5}}^{2}{\it 
\tilde{W}_6}-{\frac {67}{20}}\,{{ \tilde{W}_2}}^{2}{ \tilde{W}_3}\,{{ \tilde{W}_4}}^{3}{ \tilde{W}_5}\,
{ \tilde{W}_6}-{\frac {115}{24}}\,{ \tilde{W}_2}\,{{ \tilde{W}_3}}^{3}{{ \tilde{W}_4}}^{2}{
 \tilde{W}_5}\,{ \tilde{W}_6}& \nonumber \\ 
&+{\frac {87}{10}}\,{ \tilde{W}_2}\,{ \tilde{W}_3}\,{{ \tilde{W}_4}}^{2}
{ \tilde{W}_5}\,{{ \tilde{W}_6}}^{2}+{\frac {9}{10}}\,{ \tilde{W}_2}\,{ \tilde{W}_3}\,{ \tilde{W}_4}
\,{{ \tilde{W}_5}}^{3}{ \tilde{W}_6}-{\frac {151}{864}}\,{{ \tilde{W}_2}}^{8}{{ \tilde{W}_3}}^
{2}{{ \tilde{W}_4}}^{2}-3\,{ \tilde{W}_2}\,{{ \tilde{W}_5}}^{2}{{ \tilde{W}_6}}^{3}& \nonumber \\ 
&-\frac32  \,{\it 
\tilde{W}_4}\,{{ \tilde{W}_5}}^{4}{ \tilde{W}_6}-{\frac {33}{2}}\,{ \tilde{W}_2}\,{ \tilde{W}_4}\,{{
 \tilde{W}_6}}^{4}+{\frac {28}{25}}\,{{ \tilde{W}_2}}^{2}{{ \tilde{W}_3}}^{2}{{ \tilde{W}_5}}^
{4}+{\frac {1379}{4800}}\,{{ \tilde{W}_2}}^{4}{{ \tilde{W}_3}}^{4}{{ \tilde{W}_5}}^{2}+\frac{1}{10} \,{{ \tilde{W}_3}}^{3}{{ \tilde{W}_4}}^{4}{ \tilde{W}_5}& \nonumber \\ 
&-{\frac {391}{36}}\,{{ \tilde{W}_2}
}^{3}{{ \tilde{W}_4}}^{3}{{ \tilde{W}_6}}^{2}+{\frac {17}{12}}\,{{ \tilde{W}_3}}^{2}{{
 \tilde{W}_4}}^{3}{{ \tilde{W}_6}}^{2}-{\frac {71}{2400}}\,{{ \tilde{W}_2}}^{6}{{ \tilde{W}_4}
}^{2}{{ \tilde{W}_5}}^{2}-{\frac {55}{12}}\,{{ \tilde{W}_2}}^{4}{ \tilde{W}_4}\,{{ \tilde{W}_6
}}^{3}-{\frac {17}{36}}\,{{ \tilde{W}_2}}^{7}{ \tilde{W}_4}\,{{ \tilde{W}_6}}^{2}& \nonumber \\ 
&-{
\frac {19}{15}}\,{{ \tilde{W}_3}}^{3}{{ \tilde{W}_5}}^{3}{ \tilde{W}_6}+2\,{ \tilde{W}_3}\,{{
 \tilde{W}_5}}^{3}{{ \tilde{W}_6}}^{2}+3\,{{ \tilde{W}_4}}^{2}{{ \tilde{W}_5}}^{2}{{ \tilde{W}_6}}^{
2}+{\frac {2375}{10368}}\,{{ \tilde{W}_2}}^{4}{{ \tilde{W}_3}}^{6}{ \tilde{W}_4}-{\frac 
{27}{800}}\,{{ \tilde{W}_2}}^{7}{{ \tilde{W}_5}}^{2}{ \tilde{W}_6}& \nonumber \\ 
&+{\frac {661}{864}}\,{
{ \tilde{W}_2}}^{6}{{ \tilde{W}_3}}^{2}{{ \tilde{W}_4}}^{3}-{\frac {55}{36}}\,{{ \tilde{W}_2}}
^{6}{{ \tilde{W}_4}}^{3}{ \tilde{W}_6}-{\frac {46}{75}}\,{{ \tilde{W}_2}}^{4}{{ \tilde{W}_5}}^
{2}{{ \tilde{W}_6}}^{2}+{\frac {1}{120}}\,{{ \tilde{W}_2}}^{8}{ \tilde{W}_4}\,{{ \tilde{W}_5}}
^{2}-{\frac {56}{375}}\,{{ \tilde{W}_2}}^{6}{ \tilde{W}_3}\,{{ \tilde{W}_5}}^{3}& \nonumber \\ 
&+{\frac 
{5213}{20736}}\,{{ \tilde{W}_2}}^{6}{{ \tilde{W}_3}}^{4}{ \tilde{W}_6}+{\frac {41}{40}}
\,{ \tilde{W}_2}\,{{ \tilde{W}_4}}^{2}{{ \tilde{W}_5}}^{4}+{\frac {1267}{9600}}\,{{\it 
\tilde{W}_2}}^{7}{{ \tilde{W}_3}}^{2}{{ \tilde{W}_5}}^{2}+{\frac {73}{400}}\,{{ \tilde{W}_3}}^{4}
{{ \tilde{W}_4}}^{2}{{ \tilde{W}_5}}^{2}+{\frac {1091}{1296}}\,{{ \tilde{W}_2}}^{5}{{
 \tilde{W}_3}}^{4}{{ \tilde{W}_4}}^{2}& \nonumber \\ 
&-{\frac {37}{576}}\,{{ \tilde{W}_2}}^{8}{{ \tilde{W}_3}}
^{3}{ \tilde{W}_5}-{\frac {25}{24}}\,{{ \tilde{W}_2}}^{2}{{ \tilde{W}_4}}^{5}{ \tilde{W}_6}+{
\frac {19}{180}}\,{{ \tilde{W}_3}}^{7}{ \tilde{W}_4}\,{ \tilde{W}_5}+{\frac {13}{20}}\,{
{ \tilde{W}_2}}^{2}{{ \tilde{W}_5}}^{4}{ \tilde{W}_6}+\frac25 \,{{ \tilde{W}_3}}^{2}{ \tilde{W}_4}\,{{
 \tilde{W}_5}}^{4}& \nonumber \\ 
&+{\frac {119}{6912}}\,{{ \tilde{W}_2}}^{10}{{ \tilde{W}_3}}^{2}{ \tilde{W}_4
}-{\frac {79}{24}}\,{{ \tilde{W}_2}}^{3}{{ \tilde{W}_3}}^{2}{{ \tilde{W}_6}}^{3}-{\frac 
{7801}{41472}}\,{{ \tilde{W}_2}}^{7}{{ \tilde{W}_3}}^{4}{ \tilde{W}_4}-{\frac {215}{6912
}}\,{{ \tilde{W}_2}}^{9}{{ \tilde{W}_3}}^{2}{ \tilde{W}_6}& \nonumber \\ 
&+{\frac {13}{11520}}\,{{ \tilde{W}_2
}}^{11}{ \tilde{W}_3}\,{ \tilde{W}_5}-\frac45 \,{ \tilde{W}_2}\,{ \tilde{W}_3}\,{{ \tilde{W}_5}}^{5}+\frac74 \,{ \tilde{W}_2}\,{{ \tilde{W}_4}}^{4}{{ \tilde{W}_6}}^{2}-{\frac {175}{288}}\,{{ \tilde{W}_3}
}^{4}{{ \tilde{W}_4}}^{3}{ \tilde{W}_6}-{\frac {247}{192}}\,{{ \tilde{W}_2}}^{4}{{ \tilde{W}_3
}}^{2}{{ \tilde{W}_4}}^{4}& \nonumber \\ 
&-{\frac {31}{50}}\,{{ \tilde{W}_2}}^{3}{ \tilde{W}_4}\,{{\it 
\tilde{W}_5}}^{4}+\frac{1}{15} \,{ \tilde{W}_3}\,{{ \tilde{W}_4}}^{3}{{ \tilde{W}_5}}^{3}+{\frac {37}{2}}
\,{{ \tilde{W}_2}}^{2}{{ \tilde{W}_4}}^{2}{{ \tilde{W}_6}}^{3}+{\frac {259}{200}}\,{{
 \tilde{W}_2}}^{2}{{ \tilde{W}_4}}^{4}{{ \tilde{W}_5}}^{2}-{\frac {25}{1152}}\,{{ \tilde{W}_2}
}^{10}{ \tilde{W}_4}\,{ \tilde{W}_6}& \nonumber \\ 
&-{\frac {221}{432}}\,{{ \tilde{W}_2}}^{6}{{ \tilde{W}_3}}^
{2}{{ \tilde{W}_6}}^{2}+{\frac {19}{12960}}\,{{ \tilde{W}_2}}^{5}{{ \tilde{W}_3}}^{5}{
 \tilde{W}_5}+{\frac {313}{1152}}\,{{ \tilde{W}_2}}^{8}{{ \tilde{W}_4}}^{2}{ \tilde{W}_6}-{
\frac {161}{648}}\,{{ \tilde{W}_2}}^{3}{{ \tilde{W}_3}}^{6}{ \tilde{W}_6}-{\frac {14479}
{10368}}\,{{ \tilde{W}_2}}^{3}{{ \tilde{W}_3}}^{4}{{ \tilde{W}_4}}^{3}& \nonumber \\ 
&+\frac{1}{24} \,{{ \tilde{W}_2}}^
{2}{{ \tilde{W}_3}}^{7}{ \tilde{W}_5}-{\frac {959}{1728}}\,{{ \tilde{W}_2}}^{2}{{ \tilde{W}_3}
}^{6}{{ \tilde{W}_4}}^{2}-\frac{1}{12} \,{ \tilde{W}_2}\,{{ \tilde{W}_3}}^{8}{ \tilde{W}_4}+{\frac {35}
{9}}\,{{ \tilde{W}_2}}^{5}{{ \tilde{W}_4}}^{2}{{ \tilde{W}_6}}^{2}+{\frac {109}{32}}\,{{
 \tilde{W}_2}}^{4}{{ \tilde{W}_4}}^{4}{ \tilde{W}_6}& \nonumber \\ 
&-{\frac {99}{400}}\,{{ \tilde{W}_2}}^{4}{{
 \tilde{W}_4}}^{3}{{ \tilde{W}_5}}^{2}+{\frac {325}{216}}\,{{ \tilde{W}_2}}^{3}{{ \tilde{W}_3}
}^{4}{{ \tilde{W}_6}}^{2}-{\frac {11083}{13500}}\,{{ \tilde{W}_2}}^{3}{{ \tilde{W}_3}}^{
3}{{ \tilde{W}_5}}^{3}+{\frac {29}{72}}\,{{ \tilde{W}_2}}^{2}{{ \tilde{W}_3}}^{2}{{\it 
\tilde{W}_4}}^{5}-{\frac {19}{90}}\,{ \tilde{W}_2}\,{{ \tilde{W}_3}}^{6}{{ \tilde{W}_5}}^{2}& \nonumber \\ 
&+{
\frac {373}{1152}}\,{ \tilde{W}_2}\,{{ \tilde{W}_3}}^{4}{{ \tilde{W}_4}}^{4}+{\frac {1}{
110592}}\,{{ \tilde{W}_2}}^{15}-{\frac {1}{324}}\,{{ \tilde{W}_3}}^{10}+\frac15 \,{{
 \tilde{W}_5}}^{6}+6\,{{ \tilde{W}_6}}^{5}+{\frac {2}{27}}\,{{ \tilde{W}_3}}^{8}{ \tilde{W}_6}
+{\frac {71}{864}}\,{{ \tilde{W}_3}}^{6}{{ \tilde{W}_4}}^{3}& \nonumber \\ 
&-\frac{1}{32} \,{ \tilde{W}_2}\,{{
 \tilde{W}_4}}^{7}-{\frac {25}{36}}\,{{ \tilde{W}_3}}^{6}{{ \tilde{W}_6}}^{2}+{\frac {
229}{1125}}\,{{ \tilde{W}_3}}^{5}{{ \tilde{W}_5}}^{3}-{\frac {1}{72}}\,{{ \tilde{W}_3}}^
{2}{{ \tilde{W}_4}}^{6}+{\frac {25}{12}}\,{{ \tilde{W}_2}}^{3}{{ \tilde{W}_6}}^{4}+{
\frac {19}{6}}\,{{ \tilde{W}_3}}^{4}{{ \tilde{W}_6}}^{3}& \nonumber \\ 
&+\frac{1}{24}\,{{ \tilde{W}_4}}^{6}{\it 
\tilde{W}_6}-{\frac {1}{100}}\,{{ \tilde{W}_4}}^{5}{{ \tilde{W}_5}}^{2}-7\,{{ \tilde{W}_3}}^{2}{{
 \tilde{W}_6}}^{4}-{{ \tilde{W}_4}}^{3}{{ \tilde{W}_6}}^{3}-{\frac {7}{18432}}\,{{ \tilde{W}_2
}}^{13}{ \tilde{W}_4}-{\frac {17}{27648}}\,{{ \tilde{W}_2}}^{12}{{ \tilde{W}_3}}^{2}& \nonumber \\ 
&+{
\frac {1}{1536}}\,{{ \tilde{W}_2}}^{12}{ \tilde{W}_6}+{\frac {59}{9216}}\,{{ \tilde{W}_2
}}^{11}{{ \tilde{W}_4}}^{2}+{\frac {2339}{165888}}\,{{ \tilde{W}_2}}^{9}{{ \tilde{W}_3}}
^{4}-{\frac {1}{2400}}\,{{ \tilde{W}_2}}^{10}{{ \tilde{W}_5}}^{2}-{\frac {251}{
4608}}\,{{ \tilde{W}_2}}^{9}{{ \tilde{W}_4}}^{3}& \nonumber \\ 
&-{\frac {3779}{186624}}\,{{ \tilde{W}_2}
}^{6}{{ \tilde{W}_3}}^{6}+{\frac {11}{576}}\,{{ \tilde{W}_2}}^{9}{{ \tilde{W}_6}}^{2}+{
\frac {61}{256}}\,{{ \tilde{W}_2}}^{7}{{ \tilde{W}_4}}^{4}+{\frac {19}{1944}}\,{{
 \tilde{W}_2}}^{3}{{ \tilde{W}_3}}^{8}-{\frac {179}{384}}\,{{ \tilde{W}_2}}^{5}{{ \tilde{W}_4}
}^{5}+{\frac {61}{216}}\,{{ \tilde{W}_2}}^{6}{{ \tilde{W}_6}}^{3}& \nonumber \\ 
&+{\frac {1467}{
20000}}\,{{ \tilde{W}_2}}^{5}{{ \tilde{W}_5}}^{4}+{\frac {43}{192}}\,{{ \tilde{W}_2}}^{3
}{{ \tilde{W}_4}}^{6}. & \nonumber \eeq

\newpage

\glsaddall
%\printglossary[type=\acronymtype,title=Abbreviations]
\printnoidxglossary[type=\acronymtype,title={Appendix E: Abbreviations}]
%\printnoidxglossary[type=\acronymtype]
\newpage

\section*{Acknowledgements} 

The work was supported by NSERC of Canada (A.C.), and through the Research Council of Norway, Toppforsk grant no. 250367: Pseudo- Riemannian Geometry and Polynomial Curvature Invariants: Classification, Characterisation and Applications (D.M.).

% aabbcc
%\bibliographystyle{h-elsevier}
%\bibliographystyle{h-physrev}
%\bibliographystyle{unsrtnatalt}
\bibliographystyle{unsrt-phys}
\bibliography{GHReferences}

\end{document}